\documentclass[a4paper, authoryear, 12pt]{elsarticle}
\usepackage{lipsum}
\makeatletter
\def\ps@pprintTitle{%
	\let\@oddhead\@empty
	\let\@evenhead\@empty
	\def\@oddfoot{}%
	\let\@evenfoot\@oddfoot}
\makeatother

\usepackage[english]{babel}
\usepackage{babel}
\usepackage{fontenc}
\usepackage{graphicx}
\usepackage{caption}
\usepackage[outline]{contour}
\usepackage{amsmath,amsthm,amssymb}
\usepackage{natbib}
\usepackage[colorlinks=true,citecolor=blue]{hyperref}
\usepackage{color}
\usepackage{graphicx}
\usepackage{subfigure}
\usepackage{color}
\usepackage{newlfont}
\usepackage{multirow}
\usepackage{longtable} 
\usepackage{ifthen}
\usepackage{alltt}
\usepackage{enumerate}
\usepackage[margin=2cm]{geometry}
\usepackage{tikz}
\usepackage{epstopdf}
\usepackage{float}
\usepackage{arydshln}

\usepackage{tikz}
\usetikzlibrary{shapes,arrows}
\numberwithin{equation}{section}
\newtheorem{theorem}{Theorem}[section]
\newtheorem{lemma}{Lemma}[section]
\newtheorem{proposition}{\bf Proposition}[section]

\usepackage{lineno}
\journal{Elsevier}

\begin{document}
	\title{Effective lockdown and role of hospital-based COVID-19 transmission in some Indian states: An outbreak risk analysis}

	\author[DAC]{Tridip Sardar \footnote{Email: tridipiitk@gmail.com}}
	\author[VB]{Sourav Rana \footnote{Corresponding author. Email: sourav.rana@visva-bharati.ac.in} }
	\address[DAC]{Department of Mathematics, Dinabandhu Andrews College, Kolkata, West Bengal, India}
	\address[VB]{Department of Statistics, Visva-Bharati University, Santiniketan, West Bengal, India}

\begin{abstract}

There are several reports in India that indicate hospitals and quarantined centers are COVID-19 hotspots. In the absence of efficient contact tracing tools, Govt. and the policymakers may not be paying attention to the risk of hospital-based transmission. To explore more on this important route and its possible impact on lockdown effect, we developed a mechanistic model with hospital-based transmission. Using daily notified COVID-19 cases from six states (Maharashtra, Delhi, Madhya Pradesh, Rajasthan, Gujarat, and Uttar Pradesh) and overall India, we estimated several important parameters of the model. Moreover, we provided an estimation of the basic ($R_{0}$), the community ($R_{C}$), and the hospital ($R_{H}$) reproduction numbers for those seven locations. To obtain a reliable forecast of future COVID-19 cases, a BMA post-processing technique is used to ensemble the mechanistic model with a hybrid statistical model. Using the ensemble model, we forecast COVID-19 notified cases (daily and cumulative) from May 3, 2020, till May 20, 2020, under five different lockdown scenarios in the mentioned locations. Our analysis of the mechanistic model suggests that most of the new COVID-19 cases are currently undetected in the mentioned seven locations. Furthermore, a global sensitivity analysis of four epidemiologically measurable \& controllable parameters on $R_{0}$ and as well on the lockdown effect, indicate that if appropriate preventive measures are not taken immediately, a much larger COVID-19 outbreak may trigger from hospitals and quarantined centers. In most of the locations, our ensemble model forecast indicates a substantial percentage of increase in the COVID-19 notified cases in the coming weeks in India. Based on our results, we proposed a containment policy that may reduce the threat of a larger COVID-19 outbreak in the coming days.

\end{abstract}


\maketitle

\section*{Introduction}
Coronavirus disease of 2019 (COVID-19) was first observed in Wuhan, China and rapidly spread across the globe in a short duration~\cite{Chen2020Lancet}. World Health Organization (WHO) declares COVID-19 as pandemic after assessing its various characteristics~\cite{WHO_Pandemic}. As of April 29, 2020--over three million cases and over two hundred thousand deaths due to COVID-19 is reported across the globe~\cite{worldcovid2020track}. In India, the first confirmed case of COVID-19 was reported on $30/01/2020$~\cite{WHO2020_1}, a student from Kerala studying in a university in Wuhan \cite{India_Today5}. As of April 29, 2020--$33065$ confirmed cases and $1079$ deaths due to COVID-19 are reported in India~\cite{indiacovid2020track}.    

According to a daily monitoring report published by WHO, $22,073$~healthcare workers across $52$ countries are being tested positive to COVID-19~\cite{WHOreportheathcare2020}. The report also noted that the number provided may be an underestimation as there is no systematic reporting of infections among the healthcare workers~\cite{WHOreportheathcare2020}. In India, there are several reports that indicate hospitals and quarantined centers are COVID-19 hotspots~\cite{TheHINDUsuperspreaders, India_Today1, Economic_Times1, Print_In1, NDTV2}. The doctors, nurses and other health workers are mainly vulnerable as they are in close proximity with the COVID-19 patients~\cite{India_Today1, Economic_Times1, Print_In1, NDTV2, NDTVheathcareworkers}. The close relatives of the notified COVID-19 patients in the quarantine centers may also be at risk of getting infection. In addition, those journalists who are continuously visiting  hospitals and quarantine centers to get updated reports on COVID-19, may also be at risk of getting infection \cite{IndiaTVjournalistscovid19, Scroll_In1}. Therefore, a significant percentage of susceptible population in the community may be exposed to COVID-19 infection occurring from the contacts with patents in hospitals and quarantine centers. Due to unavailability of efficient contact tracing tools, Govt. and the policy makers may be ignoring this important transmission route of COVID-19.  

Currently there is no vaccine and effective medicine available for COVID-19. Therefore, to break the transmission chain of COVID-19, Govt. had implemented a full nation-wise lockdown (home quarantined the community) staring from March 25, 2020 till April, 14, 2020. However, the large country like India with such diverse and huge population, lockdown all over the nation may not be a very feasible and effective solution. In addition, lockdown already have a huge impact on the Indian economy specially on the short scale industries~\cite{Economic_Times2, India_Today4}. To partially overcome this economic crisis as well reducing COVID-19 transmission, Govt. has proposed some amendments (known as cluster containment strategy) on the lockdown rules from April 20, 2020~\cite{Economictimes2020a, financialexpress20a}. In these revised rules, Govt. has provided some relaxation in current rules by dividing different districts of the various states into three  zones namely red (hotspot), orange (limited human movement), and green (economic activity) depending on the number of COVID-19 cases~\cite{Economictimes2020a, financialexpress20a, businesstoday2020a}. However, question remains whether this cluster containment strategy may be successful in reducing COVID-19 transmission or not? If not then what may be other alternative solutions to reduce COVID-19 transmission? These question can only be answered by studying the dynamics and prediction of a mechanistic mathematical model for COVID-19 transmission and testing the results in real situation. 

Mechanistic mathematical models based on system of ordinary differential equations~(ODE) may provide useful information regarding transmission dynamics of COVID-19 and its control. Recently, there are several modeling studies that provide information on different effective control measures for reducing COVID-19 transmission~\cite{Moghadas2020PNAS, Tang2020, sardar2020assessment}. In their study, Moghdas et al~\cite{Moghadas2020PNAS} developed an age-structure model on COVID-19 and projected required ICU beds for different outbreak scenarios in USA. To improve the hospital capacity and possible containment of COVID-19, they recommended self isolation and better hygiene practices in the community. Using a mechanistic model, Tang et al~\cite{Tang2020} estimated the control reproduction numbers and studied the effect of various interventions on COVID-19 transmission in China. Recently, Sardar et al~\cite{sardar2020assessment} provided an effective lockdown strategy to control COVID-19 transmission in India. They recommended cluster specific lockdown policy in different parts of India for effective reduction in transmission of COVID-19. There are some recent studies that uses different statistical modeling techniques to provide reliable real-time forecast of COVID-19 cases~\cite{Chakraborty2020Arxiv,Tandon2020Arxiv,Hu2020Arxiv}. Therefore, a combination of mathematical and statistical models may be effective in determining  the control measures as well providing a robust real-time forecast of future COVID-19 cases and deaths in different parts of India and for other countries as well.

In this paper, we formulated a mechanistic model with hospital-based transmission. We assume that COVID-19 patients from the hospitals and quarantine centers can only be in contact with a small fraction of the susceptible population from the community. Furthermore, we assume different transmission rates for the community and the hospital-based infection. In the mechanistic model, we have incorporated  the lockdown effect through home quarantine of a certain percentage of susceptible population from the community. Using daily notified COVID-19 cases from six states (Maharashtra, Delhi, Madhya Pradesh, Rajasthan, Gujarat, and Uttar Pradesh) and overall India, we estimated several important parameters of the mechanistic model. Furthermore, we estimated the basic ($R_{0}$), the community ($R_{C}$), and the hospital ($R_{H}$) reproduction numbers for the seven locations under study. To obtain a reliable forecast of future COVID-19 notified cases in the above mentioned locations, we used a hybrid statistical model that can efficiently capture fluctuations in the daily time series data. BMA post-processing technique based on DRAM algorithm is used to ensemble our mechanistic mathematical model with the hybrid statistical model. Using the ensemble model, we forecast COVID-19 notified cases (daily and cumulative) from May 3, 2020 to May 20, 2020, under five different lockdown scenarios in the seven locations. To determine an effective lockdown policy, we carried out a global sensitivity analysis of four epidemiologically measurable \& controllable parameters on the lockdown effect (number of cases reduction) and as well on $R_{0}$.

\section*{Method}

We extend a previous mechanistic model~\cite{sardar2020assessment} by considering hospital-based COVID-19 transmission (see Fig~\ref{Fig:Flow_India_covid} and supplementary). We assume that hospitalized \& notified infected population can only be in contact with a small fraction ($\rho$) of the susceptible population from the community~(see Table~\ref{tab:mod1} and supplementary). We assume different transmission rates ($\beta_{1}$ and $\beta_{2}$, respectively) for community and hospital-based infection. As it is very difficult to detect asymptomatic infected in the community therefore, we assume that only a fraction of symptomatic infected population being notified \& hospitalized by COVID-19 testing at a rate, $\tau$ (see Fig~\ref{Fig:Flow_India_covid} and Table~\ref{tab:mod1}). Following~\cite{sardar2020assessment}, the disease related deaths are considered only for the notified \& hospitalized population at a rate $\delta$. We incorporated lockdown effect in our model (see Fig~\ref{Fig:Flow_India_covid} and supplementary) by home quarantined a fraction of susceptible population at a rate $l$. We also assume that after the current lockdown period $\left(\displaystyle \frac{1}{\omega} = 40 \hspace{0.1 cm} days \right)$ the home quarantined individuals will returned to the general susceptible population (see Fig~\ref{Fig:Flow_India_covid} and supplementary). Moreover, we assume that the home quarantined individuals do not mixed with the general population (see Fig~\ref{Fig:Flow_India_covid}) \textit{i.e.} this class of individuals do not contribute in the disease transmission. A flow diagram and the information on our mechanistic ODE model parameters are provided in Fig~\ref{Fig:Flow_India_covid} and Table~\ref{tab:mod1}, respectively.                      

Mechanistic ODE model (see Fig~\ref{Fig:Flow_India_covid} and supplementary) we are using for this study may be efficient in capturing overall trend of the time-series data and the transmission dynamics of COVID-19. However, as solution of the ODE model is always smooth therefore, our mechanistic model may not be able to capture fluctuations occurring in the daily time-series data. Several complex factors like lockdown, symptomatic, asymptomatic, hospital transmission, awareness, rapid testing, preventive measures, etc may influence the variations in daily COVID-19 time-series data. Therefore, it is an extremely challenging job to fit and long term forecast using this daily time-series data. To resolve this issue, we considered a Hybrid statistical model which is a combination of five forecasting models namely, Auto-regressive Integrated Moving Average model (ARIMA); Exponential smoothing state space model (ETS); Theta Method Model (THETAM); Exponential smoothing state space model with Box-Cox transformation, ARMA errors, Trend and Seasonal components (TBATS); and Neural Network Time Series Forecasts (NNETAR). Finally, the Hybrid statistical model and the mechanistic ODE model (see Fig~\ref{Fig:Flow_India_covid} and supplementary) are combined together by a post-processing Bayesian Model Averaging (BMA) technique which we discussed later in the manuscript.

We used daily confirmed COVID-19 cases from Maharashtra (MH), Delhi (DL), Madhya Pradesh (MP), Rajasthan (RJ), Gujarat (GJ), Uttar Pradesh (UP) and overall India (IND) for the time period March 14, 2020 to April 14, 2020 (MH), March 14, 2020 to April 18, 2020 (DL), March 20, 2020 to April 17, 2020 (MP), March 14, 2020 to April 18, 2020 (RJ), March 19, 2020 to April 16, 2020 (GJ),  March 14, 2020 to April 18, 2020 (UP), and March 2, 2020 to April 19, 2020 (IND) for our study. As of April 29, 2020, these referred six states contributes \textbf{74\%} of the total COVID-19 notified cases in India \cite{indiacovid2020track}. Confirmed daily COVID-19 cases form these mentioned seven locations are collected from \cite{indiacovid2020track}. State-wise population data are taken from \cite{aadhaar20}.  

We estimated several uninformative parameters (see Table~\ref{tab:mod1}) of our mechanistic model (see Fig~\ref{Fig:Flow_India_covid} and supplementary) by calibrating our mathematical model (see Fig~\ref{Fig:Flow_India_covid} and supplementary) to the daily notified COVID-19 cases from the seven locations MH, DL, MP, RJ, GJ, UP, and IND respectively. As some initial conditions of our mathematical model (see Fig~\ref{Fig:Flow_India_covid} and supplementary) are also unknown therefore, we prefer to estimate these uninformative initial conditions from the data (see Table S1 in supplementary). In lockdown 1.0, the Indian Govt. has implemented a 21 days nationwide full lockdown (home quarantined the community) starting from March 25, 2020 to April 14, 2020 \cite{Economictimes2020} and Govt. then extend the lockdown period up to May 3, 2020 (lockdown 2.0) \cite{timesofindia20b, financialexpress20}. Therefore, the daily COVID-19 time series data contains the effect of with and without lockdown scenario, therefore, we prefer to use a combination of two mathematical models (without and with lockdown) for calibration. An elaboration on the combination technique of the mathematical models without and with lockdown (see Fig~\ref{Fig:Flow_India_covid} and supplementary) is provided below:

\begin{description}
	\item[$\bullet$] We first use the mechanistic model without lockdown (see Eq~S1 in supplementary) starting from the first date of the daily COVID-19 data up to end of March 24, 2020 for the seven locations MH, DL, MP, RJ, GJ, UP, and IND, respectively. 	
	\item[$\bullet$] Using values of the state variables of the model without lockdown (see Model~S1 in supplementary method) onset of March 24, 2020 as initial conditions, we run the mechanistic model with lockdown (see Fig~\ref{Fig:Flow_India_covid} and Model~S2 in supplementary) up to the end date of the daily COVID-19 data for the seven locations MH, DL, MP, RJ, GJ, UP, and IND, respectively.	
\end{description} The nonlinear least square function $'lsqnonlin'$ in the \textit{MATLAB} based optimization toolbox is called to fit the simulated and observed daily COVID-19 notified cases in those seven locations mentioned earlier. Bayesian based $'DRAM'$ algorithm \cite{haario2006dram} is used to sample the uninformative parameters and initial conditions (see Table~S1 and Table~S2 in supplementary) of the mathematical models combination without and with lockdown (see Fig~\ref{Fig:Flow_India_covid} and supplementary). A details on mechanistic model fitting is provided in \cite{sardar2017mathematical}. 

Calibration of the Hybrid statistical model for the mentioned seven states are done using the $R$ package $'forecastHybrid'$. First, we fitted the individual models ARIMA, ETS, THETAM, TBATS, and NNETAR by calling the functions $'auto.arima'$, $'ets'$, $'thetam'$, $'tbats'$, and $'nnetar'$ respectively. The results generated from each of the above models are combined with equal weights to determine the Hybrid statistical model. Equal weight among the five individual models are taken as it generates a robust result (see Table S4 in supplementary) for the Hybrid statistical model \cite{lemke2010meta}.   

Post-processing BMA technique for combining the mechanistic model (see Fig~\ref{Fig:Flow_India_covid} and supplementary) and the Hybrid statistical model is based on $'DRAM'$ algorithm \cite{haario2006dram}. Let, $\displaystyle Y^{ODE} = \{y_{j}^{ODE}\}_{j=1}^n$ and $\displaystyle Y^{HBD} = \{y_{j}^{HBD}\}_{j=1}^n$ be $n$ simulated observations from our mechanistic ODE model (see Fig~\ref{Fig:Flow_India_covid} and supplementary) and the Hybrid statistical model respectively and let, $\displaystyle \hat{Y} = \{y_{j}^{obs}\}_{j=1}^n$ be $n$ observation from the data. Then

\begin{equation}
\displaystyle Y^{E} = w_{1} \hspace{0.2cm} Y^{ODE} + w_{2} \hspace{0.2cm} Y^{HBD},   
\label{Eq:ensemble}
\end{equation} is our ensemble model, where, the weights $w_{1}$ and $w_{2}$ satisfies the constraints
\begin{equation*}
\displaystyle \Delta = \left\lbrace w_{1}, w_{2} \geq 0: w_{1} + w_{2} =1 \right\rbrace.
\end{equation*}We assume $w_{1}$ and $w_{2}$ follows Gaussian proposal distribution. Then the error sum of square function \cite{haario2006dram} is defined as:  

\begin{equation*}
\displaystyle SS(\tilde{\theta}) = \sum_{i=1}^{n} \left( \hat{Y} - Y^{E}(\tilde{\theta})\right)^{2}.
\end{equation*} Posterior distribution of the weights $\tilde{\theta} = \left( \tilde{w_{1}}, \tilde{w_{2}} \right)$ for the ensemble model~(\ref{Eq:ensemble}) are generated using Bayesian based $'DRAM'$ algorithm \cite{haario2006dram} (see Table~S3 and Fig~S1 to S7 in supplementary).               
   
To save the countries short-scale industries and the agricultural sectors, Indian Govt. has proposed some amendments on current lockdown rules from April 20, 2020~\cite{Economictimes2020a, financialexpress20a}. In these revised rules, Govt. has provided some relaxation in current rules by dividing different districts of the various states into three red (hotspot), orange (limited human movement), and green (Economic activity) zones depending on the number of COVID-19 cases~\cite{Economictimes2020a, financialexpress20a, businesstoday2020a}. Implementation of these new rules in our mechanistic models combination~(see Fig~\ref{Fig:Flow_India_covid} and supplementary) are based on the following assumptions:

\begin{description}
	\item[$\bullet$] Lockdown rule will be relaxed from April 20, 2020 in those states where the current estimate of the lockdown rate (see Table~\ref{Tab:estimated-parameters-Table}) is higher than a threshold value. This relaxation in lockdown is based on the fact that locations where lockdown are strictly implemented before April 20, 2020 are likely to have more impact on the economic growth.   
	
	\item[$\bullet$] Lockdown rule will be more intensive from April 20, 2020 in those states where the current estimate of the lockdown rate (see Table~\ref{Tab:estimated-parameters-Table}) is below a threshold value. This assumption is made because locations where lockdown are not implemented properly before April 20, 2020 are likely to have more red (hotspot) zones. 
	
	\item[$\bullet$] 50\% lockdown success is taken as the threshold value for our study. Here, 50\% lockdown success in Delhi means that 50\% of the susceptible population in this state is successfully home-quarantined during the period March 25, 2020 till April 20, 2020. 
	
	\item[$\bullet$]  Insensitivity and relaxation in lockdown are measured in a same scale namely 10\%, 20\% and 30\% increment or decrement on the current estimate of lockdown rate (see Table~\ref{Tab:estimated-parameters-Table}).  	
\end{description}  
   
Following the above assumptions and using our ensemble model~(\ref{Eq:ensemble}), we provided a forecast of notified COVID-19 cases (daily and cumulative) for MH, DL, MP, RJ, GJ, UP, and IND, respectively during May 3, 2020 till May 20, 2020. As COVID-19 notified cases are continuously rising in these seven locations therefore, it is more likely that lockdown period will be extended beyond May 3, 2020. Therefore, forecast using the ensemble model~(\ref{Eq:ensemble}) during the mentioned time duration in those seven locations are based on the following scenarios:\\
(\textbf{A1}) Using our mechanistic models combination~(see Fig~\ref{Fig:Flow_India_covid} and supplementary) and the current estimate of the lockdown rates (see Table~\ref{Tab:estimated-parameters-Table}), we forecast notified COVID-19 cases up to May 20, 2020. This forecast is combined together with the results obtain from the hybrid statistical model by using our ensemble model~(\ref{Eq:ensemble}). \\  

(\textbf{A2}) Using our mechanistic models combination~(see Fig~\ref{Fig:Flow_India_covid} and supplementary) and the current estimate of the lockdown rates (see Table~\ref{Tab:estimated-parameters-Table}), we forecast notified COVID-19 cases up to April 20, 2020. From April 21, 2020 till May 20, 2020, forecast are made using \textbf{10\%} increment or decrement (depending on the state) in the estimate of current lockdown rate. This forecast is combined together with the results obtain from the hybrid statistical model by using our ensemble model~(\ref{Eq:ensemble}).\\ 

(\textbf{A3}) Using our mechanistic models combination~(see Fig~\ref{Fig:Flow_India_covid} and supplementary) and the current estimate of the lockdown rates (see Table~\ref{Tab:estimated-parameters-Table}), we forecast notified COVID-19 cases up to April 20, 2020. From April 21, 2020 till May 20, 2020, forecast are made using \textbf{20\%} increment or decrement (depending on the state) in the estimate of current lockdown rate. This forecast is combined together with the results obtain from the hybrid statistical model by using our ensemble model~(\ref{Eq:ensemble}).\\ 

(\textbf{A4}) Using our mechanistic models combination~(see Fig~\ref{Fig:Flow_India_covid} and supplementary) and the current estimate of the lockdown rates (see Table~\ref{Tab:estimated-parameters-Table}), we forecast notified COVID-19 cases up to April 20, 2020. From April 21, 2020 till May 20, 2020, forecast are made using \textbf{30\%} increment or decrement (depending on the state) in the estimate of current lockdown rate. This forecast is combined together with the results obtain from the hybrid statistical model by using our ensemble model~(\ref{Eq:ensemble}).\\

(\textbf{A5}) Using our mechanistic models combination~(see Fig~\ref{Fig:Flow_India_covid} and supplementary) and the current estimate of the lockdown rates (see Table~\ref{Tab:estimated-parameters-Table}), we forecast notified COVID-19 cases up to May 2, 2020. From May 3, 2020 till May 20, 2020, forecast are made with no lockdown (see Eq.~S1 in supplementary). This forecast is combined together with the results obtain from the hybrid statistical model by using our ensemble model~(\ref{Eq:ensemble}).\\

As, we assumed that lockdown individuals do not mixed with the general population therefore, the basic reproduction number ($R_{0}$) with and without lockdown (see Fig~\ref{Fig:Flow_India_covid} and supplementary) are equal~\cite{van2002reproduction}:

\begin{equation}
\displaystyle R_{0}=\frac{\beta_1 \kappa \sigma}{(\mu + \sigma)(\gamma_1 + \mu + \tau)}+\frac{\beta_1 (1-\kappa) \sigma}{(\mu+\gamma_2)(\mu + \sigma)} +\frac{\beta_2 \kappa \rho \sigma \tau}{(\mu + \sigma)(\delta + \gamma_3 + \mu)(\gamma_1 + \mu + \tau)}.\\\nonumber
\end{equation}     

In the expression of $R_{0}$, sum of first two term indicate the community infection occurring from symptomatic and asymptomatic infected population. The third term in $R_{0}$ specify the hospital-based COVID-19 transmission. To distinguish the community and hospital-based COVID-19 transmission, we defined the community reproduction number ($R_{C}$), and the hospital reproduction number ($R_{H}$) as follows:
      
\begin{equation}
\displaystyle R_{C}=\frac{\beta_1 \kappa \sigma}{(\mu + \sigma)(\gamma_1 + \mu + \tau)}+\frac{\beta_1 (1-\kappa) \sigma}{(\mu+\gamma_2)(\mu + \sigma)},\\\nonumber
\end{equation}and  
\begin{equation}
\displaystyle R_{H}=\frac{\beta_2 \hspace{0.04cm} \kappa \hspace{0.04cm} \rho \hspace{0.04cm}\sigma\hspace{0.04cm} \tau}{(\mu + \sigma)(\delta + \gamma_3 + \mu)(\gamma_1 + \mu + \tau)}. \\\nonumber
\end{equation} Using estimated values of epidemiologically uninformative parameters (see Table~\ref{Tab:estimated-parameters-Table}), we estimated $R_{0}$, $R_{C}$, and $R_{H}$ for the seven locations MH, DL, MP, RJ, GJ, UP, and IND, respectively.  

Constructing an effective policy on future lockdown in a region will require some relation between effect of lockdown (number of COVID-19 case reduction) with some important epidemiologically measurable \& controllable parameters. Our mechanistic ODE model~(see Fig~\ref{Fig:Flow_India_covid} and supplementary) has several important parameters and among them measurable and controllable parameters are $\beta_{2}$: average rate of transmission occurring from hospitalized \& notified based contacts (it can be controllable by following WHO guidelines); $\rho$: fraction of susceptible population from the community that are exposed to notified \& hospitalized based contacts (it also can be minimized by following proper guidelines from WHO); $\kappa$: fraction of infected that are symptomatic (rapid COVID-19 testing can provide an accurate estimate); $\tau$: notification \& hospitalization rate of symptomatic infected population (it also depend on number of COVID-19 testing). We perform a global sensitivity analysis~\cite{marino2008methodology} to determine the effect of these parameters on the lockdown effect and on the basic reproduction number ($R_{0}$), respectively. The effect of lockdown is measured in terms of differences in the total number of COVID-19 cases occurred during May 3, 2020 till May 20, 2020 under the lockdown scenarios (\textbf{A1}) and (\textbf{A5}), respectively. We draw $500$ samples from the biologically feasible ranges of the mentioned four parameters (see Table~\ref{tab:mod1}) using Latin Hypercube Sampling (LHS) technique. Other informative and uninformative parameters during simulation of the mechanistic model are taken from Table~\ref{tab:mod1} and Table~\ref{Tab:estimated-parameters-Table}, respectively. Partial rank correlation coefficients (PRCC) and its corresponding $p$-value are evaluated to determine the effect of these mentioned four parameters on the lockdown effect and the basic reproduction number ($R_{0}$), respectively.        
\section*{Results and Discussion} 

Testing of the three models (the mechanistic ODE model, the hybrid statistical model, and the ensemble model) on daily notified COVID-19 cases from Maharashtra (MH), Delhi (DL), Madhya Pradesh (MP), Rajasthan (RJ), Gujarat (GJ), Uttar Pradesh (UP), and India (IND) are presented in Fig.~\ref{Fig:Model-fitting}. Based on the performance on testing data from the mentioned seven locations, we estimated the weights ($w_{1}$ and $w_{2}$) for our ensemble model~(\ref{Eq:ensemble}) (see Table S2 in supplementary). Our result suggest that the mechanistic ODE model (see Fig~\ref{Fig:Flow_India_covid} and supplementary) displayed a better performance in RJ, UP, and IND compared to the hybrid statistical model (see Table S2 in supplementary). For rest of the locations (MH, DL, MP, and GJ), the hybrid statistical model has performed better than the mechanistic ODE model in terms of capturing the trend of the time-series data~(Table~S2 and Table~S3 in supplementary). The ensemble model~(\ref{Eq:ensemble}), which is derived from a combination of the mechanistic ODE~(see Fig~\ref{Fig:Flow_India_covid} and supplementary) and the hybrid statistical model respectively, has provided a robust result in all of these mentioned seven locations in terms of capturing time-series data trend (see Fig.~\ref{Fig:Model-fitting}).  

The estimates of uninformative parameters of the mechanistic ODE model (Fig~\ref{Fig:Flow_India_covid} and supplementary) suggests that currently in the seven locations (MH, DL, MP, RJ, GJ, UP, and IND) the community infection is mainly dominated by contribution from the asymptomatic infected population (Table~\ref{Tab:estimated-parameters-Table}). Among the seven locations, the lowest percentage of symptomatic infection in the community is found in~RJ (about \textbf{0.1\%}) and the highest percentage is found in IND (about \textbf{35\%}) (Table~\ref{Tab:estimated-parameters-Table}). Our estimates suggest that currently in the seven locations (MH, DL, MP, RJ, GJ, UP, and IND), the notification \& hospitalization rate of symptomatic infected population is about \textbf{0.2\%} to \textbf{23 \%} (Table~\ref{Tab:estimated-parameters-Table}). Therefore, most of the COVID-19 infections in those mentioned seven locations are currently undetected. Our result agrees with the recent report by the Indian Council of Medical Research (ICMR)~\cite{Indiatodayasymptomatic}. Our estimates of the lockdown rate for MH, DL, MP, RJ, GJ, UP, and IND, respectively, suggest that lockdown is properly implemented in the two metro cities DL and MH. Also, in overall India (IND) lockdown is properly implemented. In these three locations (MH, DL, and IND), about \textbf{61\%} to \textbf{77\%} of the total susceptible population may be successfully home quarantined during the present lockdown period (Table~\ref{Tab:estimated-parameters-Table}). However, for rest of the locations, our results suggest that lockdown may not be successful as about \textbf{11\%} to \textbf{49\%} of the total susceptible population may be isolated (home quarantined) during the current lockdown period in MP, RJ, GJ, and UP, respectively (Table~\ref{Tab:estimated-parameters-Table}). In the seven locations, we found that about \textbf{1\%} to \textbf{9\%} of the total susceptible populations may be exposed to hospital (notified \& hospitalized population) related contacts (Table~\ref{Tab:estimated-parameters-Table}). Considering the fact that estimates of the average hospital-based transmission rates for the seven locations (MH, DL, MP, RJ, GJ, UP, and IND) are very high~(Table~\ref{Tab:estimated-parameters-Table}), therefore, may be a significant amount of COVID-19 infection in these seven locations are currently occurring due to notified \& hospitalized infected related contacts. These findings can be further justified by analyzing the estimates of the basic ($R_{0}$), the community ($R_{C}$) and the hospitalized ($R_{H}$) reproduction numbers for~MH, DL, MP, RJ, GJ, UP, and IND, respectively (see Table~\ref{Tab:estimated-R0-Table}). Except for the \textbf{RJ}, in the remaining six locations, we found that about \textbf{1\%} to \textbf{16\%} of the total COVID-19 transmission currently occurring from notified \& hospitalized infected related contacts (see Table~\ref{Tab:estimated-R0-Table}). These figures can be increased up to \textbf{43\%} to \textbf{69\%} if proper measures are not taken in MH, DL, MP, GJ, UP, and IND, respectively, (see Table~\ref{Tab:estimated-R0-Table}). This is a worrisome situation as higher value of the hospital-based transmission rate in MH, DL, MP, GJ, UP, and IND, respectively (Table~\ref{Tab:estimated-parameters-Table}), indicate that there may be super-spreading incidents occurring from hospital-based contacts. In~RJ, low contribution of $R_{H}$ on $R_{0}$ (see Table~\ref{Tab:estimated-R0-Table}) may be due to existence of low percentage of the symptomatic infected population in the community (Table~\ref{Tab:estimated-parameters-Table}) and that leads to low percentage of notified \& hospitalized COVID-19 cases. For further investigation on super-spreading events, we carried a global uncertainty and sensitivity analysis of some epidemiologically measurable and controllable parameters from our mechanistic ODE model (Fig~\ref{Fig:Flow_India_covid} and supplementary) namely, $\boldsymbol{\beta_{2}}$: average rate of transmission occurring from notified \& hospitalized based contact (it can be controllable following the WHO guidelines~\cite{WHOhospitaltransmission}), $\boldsymbol{\rho}$: fraction of susceptible population that are exposed to hospital-based contact (it can be reduced by following proper guidelines from the WHO~\cite{WHOhospitaltransmission}), $\boldsymbol{\kappa}$: fraction of the newly infected that are symptomatic (Rapid COVID-19 testing can provide an accurate estimate), $\boldsymbol{\tau}$: hospitalization \& notification rate of symptomatic infected population (it also depend on number of COVID-19 testing) on the basic reproduction number ($R_{0}$). Partial rank correlation coefficients (PRCC) and its corresponding $p$-value suggested that all these four parameters have significant positive correlation with $R_{0}$ (see Fig.~\ref{Fig:Sensitivity-IND} and  Fig.~S7 to Fig.S12 in supplementary). Furthermore, high positive correlation of $\rho$ on $R_{0}$ indicate that small increase in the percentage of susceptible population from the community that are exposed hospital-based transmission will leads to significant increase in COVID-19 force of infection. Considering the fact that estimated value of $\beta_{2}$ (see Table~\ref{Tab:estimated-parameters-Table})~in MH, DL, MP, RJ, GJ, UP, and IND, respectively are very high (much higher than community transmission rate), therefore, a small increase in $\rho$ may leads to a larger COVID-19 outbreak in those seven locations. Therefore, until and otherwise any preventive measures are taken in these locations, a larger COVID-19 outbreak may trigger from hospitals and quarantine centers.

Using the ensemble model~(\ref{Eq:ensemble}), we forecast daily as well as total COVID-19 cases under five different lockdown scenarios in MH, DL, MP, RJ, GJ, UP, and IND, respectively, from May 3, 2020 to~May 20, 2020,~(see Fig.~\ref{Fig:Forecast-IND}, Table~\ref{Tab:cases-preiction-Table} and Fig.~S1 to~Fig.~S6 in supplementary). Comparing the projected total COVID-19 cases during~May 3, 2020 to~May 20, 2020, (see Table~\ref{Tab:cases-preiction-Table}) with the total observed cases~\cite{indiacovid2020track} during March 2, 2020 till April 29, 2020, we found about two fold increase in the total cases in~MH, MP, GJ, UP, and IND, respectively. In summary, our forecast result suggest that in the coming two weeks a significant increase in cases may be observed in most of these locations.

To determine which epidemiologically measurable and controllable parameters are most influencing the effect of lockdown, we carried out a global uncertainty analysis of $\beta_{2}$, $\rho$, $\kappa$, and $\tau$ on the lockdown effect. The lockdown effect is measured in terms of differences in the total number of COVID-19 cases during May 3, 2020 till May 20, 2020, in MH, DL, MP, RJ, GJ, UP, and IND, respectively, under the lockdown scenarios (\textbf{A1}) and (\textbf{A5}), respectively (see method section for details). For MH, PRCC and its corresponding $p$-value suggested that all these four parameters have significant influence on the lockdown effect (see Fig.~S7 in supplementary). Furthermore, significant negative correlation of $\beta_{2}$, and $\rho$ with the lockdown effect~(see Fig.~S7 in supplementary) suggested that only home quarantined the community may not be sufficient to reduce COVID-19 transmission in MH. Govt. and the policy makers may also have to focus on reducing the transmission occurring from hospital premises based on the guidelines from the WHO \cite{WHOhospitaltransmission}. For DL, PRCC and its corresponding $p$-value suggested that $\beta_{2}$, $\rho$, and $\kappa$ are the main parameters that are influencing the lockdown effect (see Fig.~S8 in supplementary). Moreover, significant negative correlation of $\beta_{2}$ and $\rho$ with the lockdown effect and as well as significant positive correlation of $\kappa$ with the lockdown effect~(see Fig.~S8 in supplementary) implies that an effective lockdown policy in DL may be a combination of lockdown in the community, contact tracing of COVID-19 cases, and with some effort in reducing hospital-based transmission following WHO guidelines~\cite{WHOhospitaltransmission}. For MP, PRCC and its corresponding $p$-value suggested that $\kappa$ and $\tau$ have high positive correlation with the lockdown effect~(see Fig.~S9 in supplementary). Furthermore, $\rho$ have significant negative correlation with the lockdown effect~(see Fig.~S9 in supplementary). Therefore an effective lockdown policy in MP may be a strict implementation of lockdown in the red and orange zones, rapid COVID-19 testing in the community and reducing hospital-based transmission by following guidelines from WHO~\cite{WHOhospitaltransmission}. For RJ, PRCC and its corresponding $p$-value suggested that only $\kappa$ have significant positive correlation with the lockdown effect~(see Fig.~S10 in supplementary). No significant correlation with hospital-based parameters may be due to existence of low percentage of the symptomatic infected population in the community (see Table~\ref{Tab:estimated-parameters-Table}) and that leads to low percentage of notified \& hospitalized based COVID-19 transmission. Therefore, RJ Govt. may focused more on contact tracing in the community with relaxation may be given in the Green and Orange zones to increase the percentage of symptomatic infected in the community. For GJ, PRCC and its corresponding $p$-value suggested that all these four parameters have significant influence on the lockdown effect (see Fig.~S11 in supplementary). Furthermore, significant negative correlation of $\beta_{2}$, and $\rho$ with the lockdown effect~(see Fig.~S11 in supplementary) indicate that only home quarantined the community may not be sufficient to reduce COVID-19 transmission in GJ. Govt. of GJ and policy makers may also have to focus on reducing the transmission occurring from hospital premises based on the guidelines from the WHO \cite{WHOhospitaltransmission}. For UP, PRCC and its corresponding $p$-value suggested that $\beta_{2}$, $\rho$, and $\kappa$ are the main parameters that are influencing the lockdown effect (see Fig.~S12 in supplementary). Moreover, significant negative correlation of $\beta_{2}$ and $\rho$ with the lockdown effect and as well as significant positive correlation of $\kappa$ with the lockdown effect~(see Fig.~S12 in supplementary) implies that an effective lockdown policy in UP may be a combination of lockdown (relaxation in Green zone), contact tracing in community with effort in reducing hospital-based transmission following the WHO guidelines~\cite{WHOhospitaltransmission}. Finally for IND, PRCC and its corresponding $p$-value suggested that $\beta_{2}$ and $\rho$ are the main parameters that are most influencing the lockdown effect (see Fig.~\ref{Fig:Sensitivity-IND}). Therefore, only home quarantined the community may not be sufficient to reduce COVID-19 transmission in IND. Govt. of IND and the policy makers may also have to focus on reducing the transmission occurring from hospital premises based on the guidelines from the WHO \cite{WHOhospitaltransmission}.        

\section*{Conclusion} 
Our analysis of the mechanistic model with hospital-based COVID-19 transmission suggest that most of the new infections occurring in India as well most of the states are currently undetected. Furthermore, a global sensitivity analysis of two epidemiologically controllable parameters from the hospital-based transmission on the basic reproduction nuber ($R_{0}$), indicate that if appropriate preventive measures are not taken immediately, a much larger COVID-19 outbreak may trigger due to the transmission occurring from the hospitalized \& notified based contacts. Moreover, our ensemble forecast model~(\ref{Eq:ensemble}) predicted a substantial percentage of increase in the COVID-19 notified cases during May 3, 2020 till May 20, 2020, (see Table~\ref{Tab:cases-preiction-Table}) in most of these locations. In Rajasthan, trend of the forecast data (see Fig~S4 in supplementary) during May 3, 2020 till May 20, 2020, is showing a decreasing trend. This is may be due to low number of hospitalized and reported cases in this state (see Table~\ref{Tab:estimated-parameters-Table}). However, cases may rise in Rajasthan if relaxation in lockdown is applied. Furthermore, trend of the forecast data in overall India (see Fig~\ref{Fig:Forecast-IND}) during May 3, 2020 till May 20, 2020, indicating the fact that~\textbf{reaching the peak of the COVID-19 epidemic curve may be a long way ahead for India}. Finally, based on our results of global sensitivity analysis of the four important epidemiologically measurable \& controllable parameters on the lockdown effect, we are suggesting the following policy that may reduce the threat of a larger COVID-19 outbreak in the coming days:\\

\textbf{Effective Lockdown Policy:} Dividing different states into three clusters (red, orange, and green) is well appreciated as it increases the percentage of symptomatic infection in the community. However, more COVID-19 testing is needed as it increases the number of notified \& hospitalized cases over the states. It is much easier to reduce hospital-based transmission in compare to community transmission. To reduce the hospitalized \& notified based contacts, an efficient disaster management team is required. They will continuously monitor the situations in different hospitals and quarantine centers across India. This team must ensure that proper safety measures are being followed based on the guidelines provided by ICMR and WHO~\cite{WHOhospitaltransmission}.

\section*{Conflict of interests}
The authors declare that they have no conflicts of interest.
\section*{Acknowledgments}
Dr. Tridip Sardar acknowledges the Science \& Engineering Research Board (SERB) major project grant  (File No: EEQ/2019/000008 dt. 4/11/2019), Govt. of India.\\\\
The Funder had no role in study design, data collection and analysis, decision to publish, or preparation of the manuscript.

\clearpage
\bibliographystyle{ieeetr}
\biboptions{square}
\bibliography{covid19_india}
\clearpage

\begin{center}
	\section*{\Large{\underline{Figures}}}
\end{center}

\begin{figure}[ht]
	\captionsetup{width=1.1\textwidth}
		\includegraphics[width=1.0\textwidth]{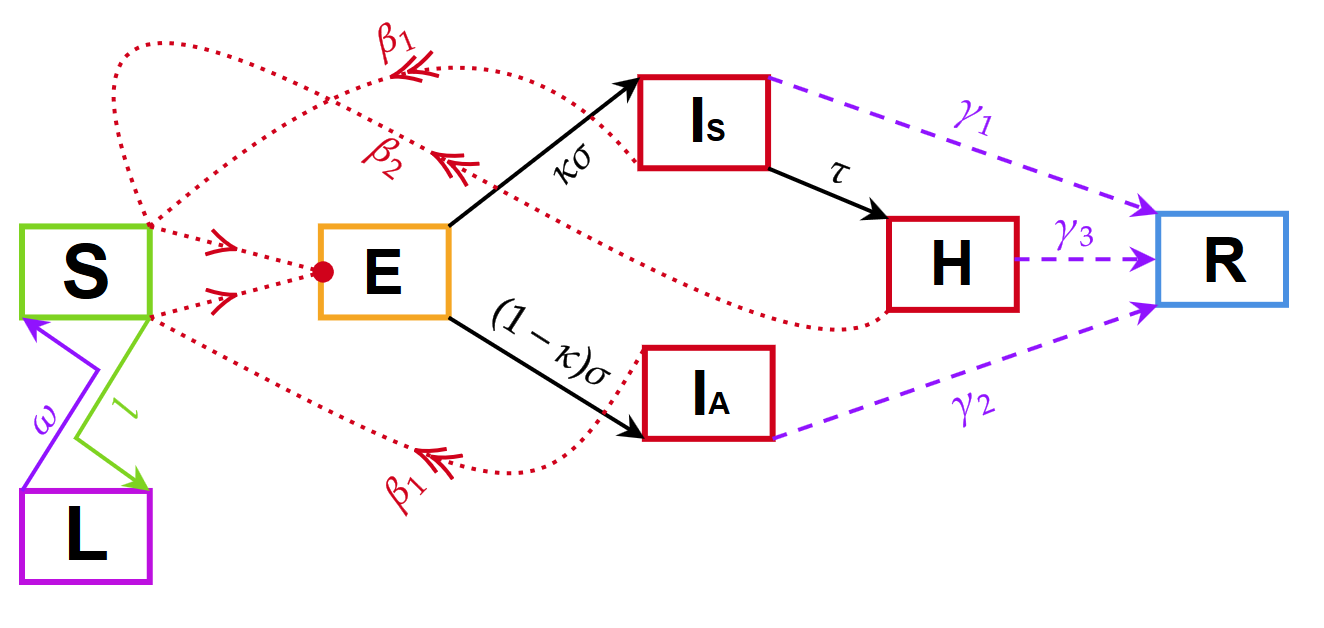}
	\caption{A Flow diagram of the mechanistic ODE model with hospital-based COVID-19 transmission and lockdown effect. Different class of population shown in this figure are $S$:~Susceptible population, $E$:~Exposed population, $I_{S}$:~COVID-19 symptomatic infected population, $I_{A}$:~COVID-19 asymptomatic infected population, $H$:~Notified \& Hospital individuals suffering from COVID-19 infection, $R$:~COVID-19 recovered population, and $L$:~Home quarantined susceptible population during lockdown, respectively. Epidemiological information of the parameters shown in this figure are provided in Table~\ref{tab:mod1}.}
	\label{Fig:Flow_India_covid}
\end{figure}

\begin{figure}[ht]
	\captionsetup{width=1.1\textwidth}
	{\includegraphics[width=1.15\textwidth]{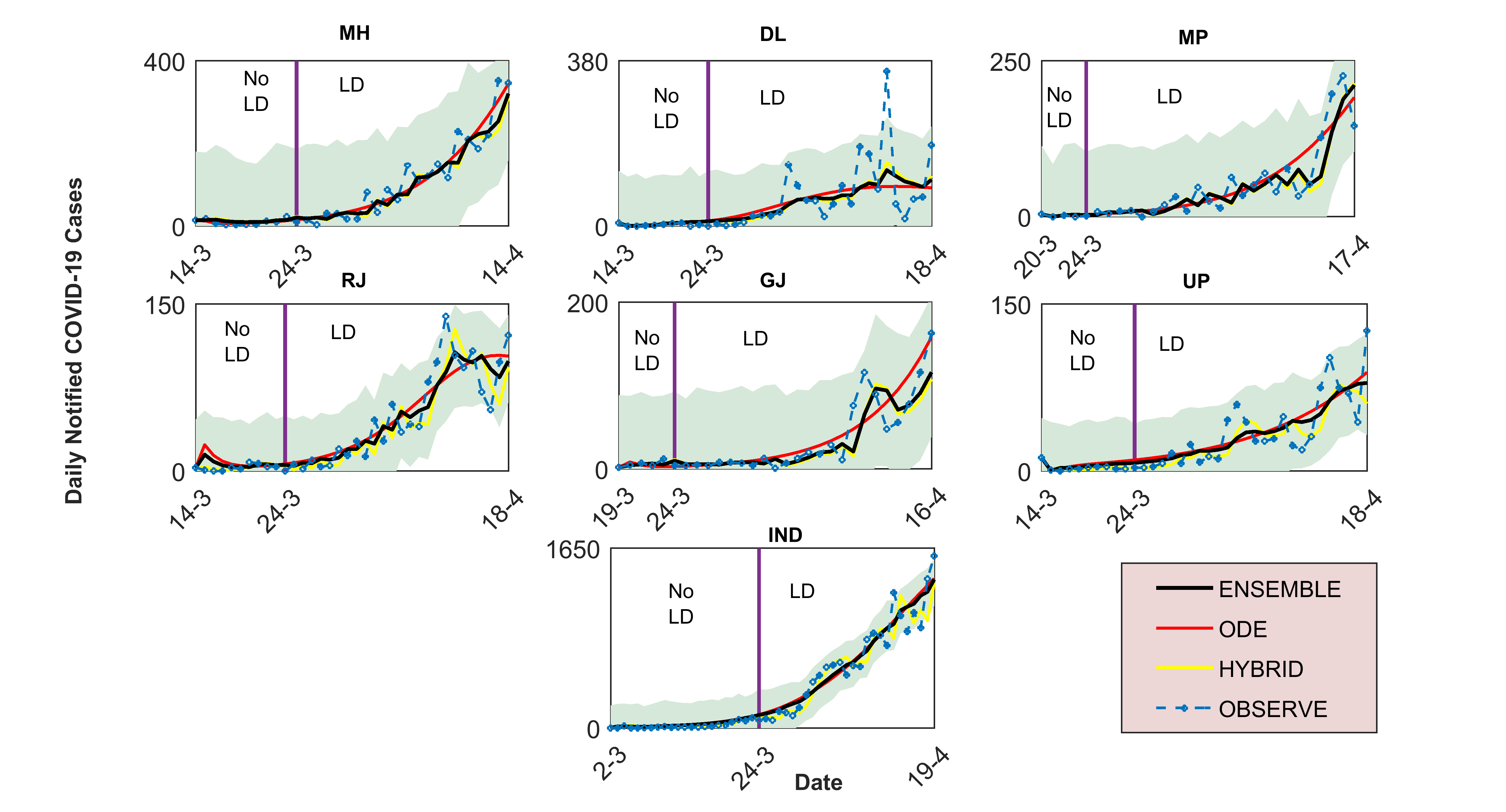}}
	\caption{Combination of mechanistic ODE models~(Fig~\ref{Fig:Flow_India_covid} and supplementary), Hybrid statistical model, and Ensemble model~(\ref{Eq:ensemble}) fitting to the daily notified COVID-19 cases form six different states and overall India. Respective subscripts are \textbf{MH}: Maharashtra, \textbf{DL}: Delhi, \textbf{MP}: Madhya Pradesh, \textbf{RJ}: Rajasthan, \textbf{GJ}: Gujarat, \textbf{UP}: Uttar Pradesh, and \textbf{IND}: India. Here,~\textbf{LD} indicate the period after the lockdown implementation in overall India started at 25/03/2020 and \textbf{No LD} specifies the period before lockdown. Lockdown effect is considered only for the mechanistic ODE model (Fig~\ref{Fig:Flow_India_covid} and supplementary) and consequently in the ensemble model~(\ref{Eq:ensemble}). Shaded area specifies the 95\% confidence region.}     
\label{Fig:Model-fitting}
\end{figure}

\begin{figure}[ht]
	\captionsetup{width=1.1\textwidth}
	{\includegraphics[width=1.15\textwidth]{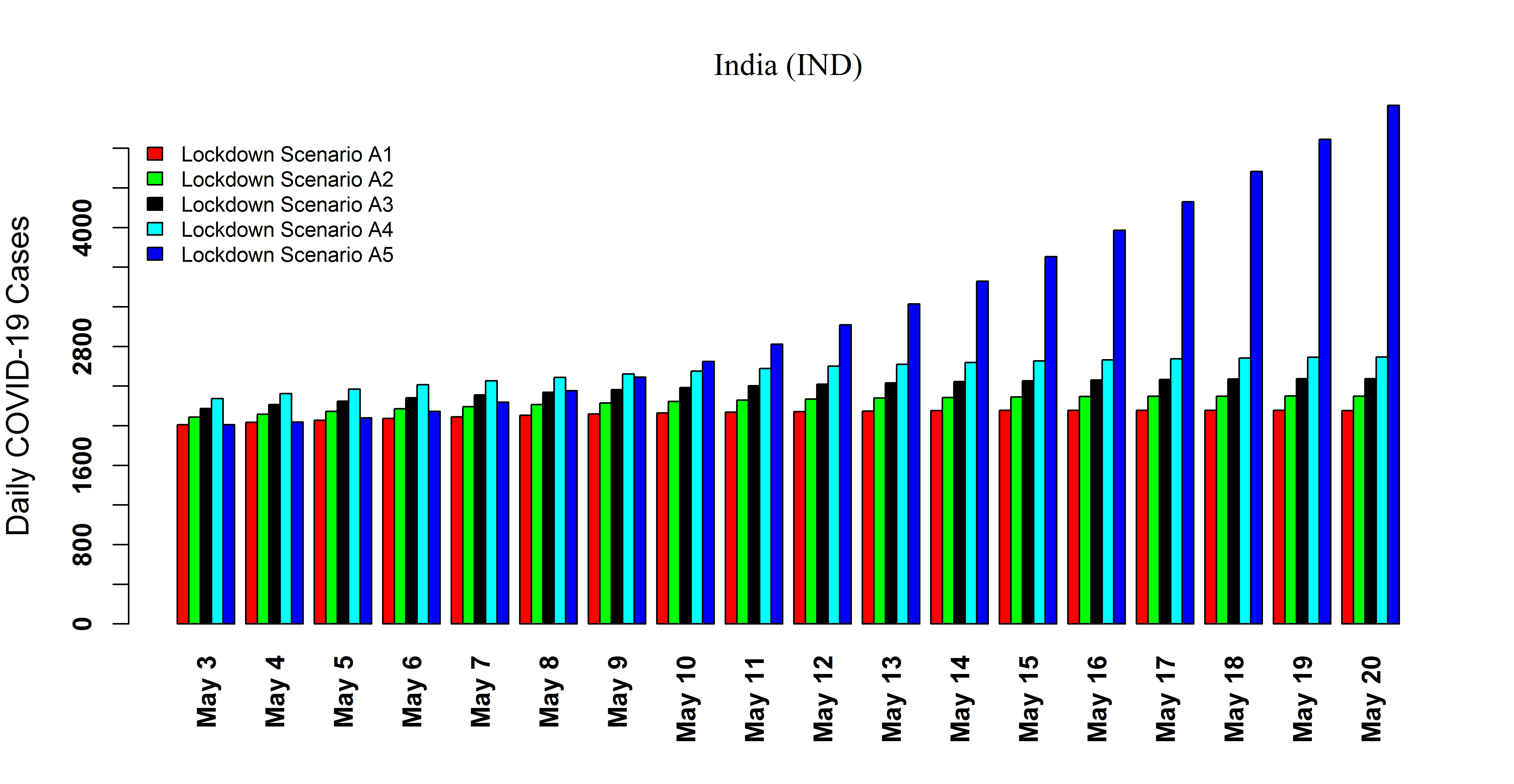}}
	\caption{Ensemble model~(\ref{Eq:ensemble}) forecast for the daily notified COVID-19 cases in India during May 3, 2020 till May 20, 2020 under five different lockdown scenarios. Under lockdown scenarios (\textbf{A2}) to (\textbf{A4}), projections are made with 10\%, 20\% and 30\% \textbf{decrement} in the current estimate of the lockdown rate, respectively (see Table~\ref{Tab:estimated-parameters-Table}). In lockdown scenario (\textbf{A1}), projections are made with current estimated lockdown rate (see Table~\ref{Tab:estimated-parameters-Table}). Finally, in the lockdown scenario (\textbf{A5}), projections are made with no lockdown during May 3, 2020 till May 20, 2020.}     
	\label{Fig:Forecast-IND}
\end{figure}

\begin{figure}[ht]
	\captionsetup{width=1.1\textwidth}
	{\includegraphics[width=1.15\textwidth]{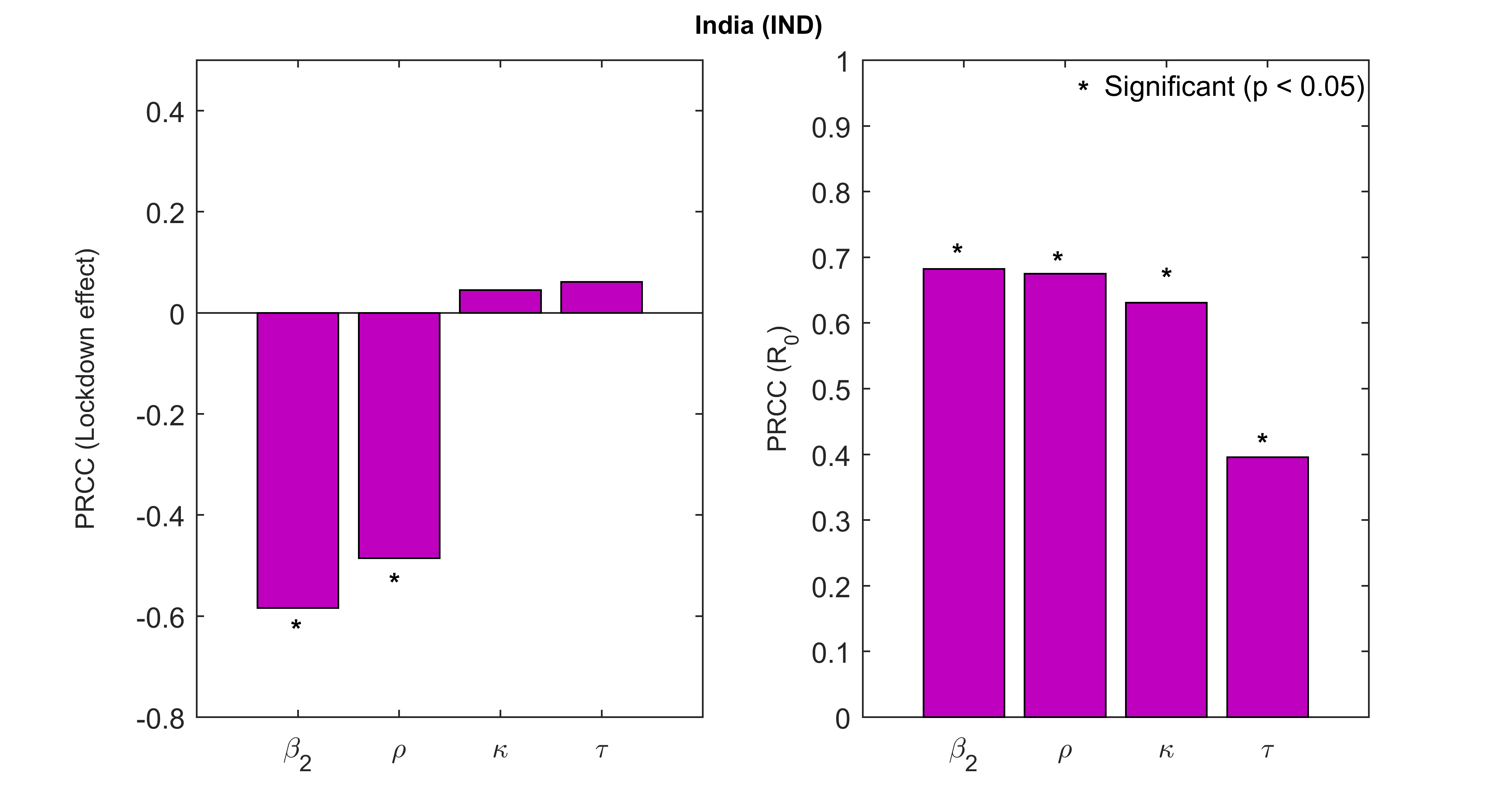}}
	\caption{Effect of uncertainty of four epidemiologically measurable and controllable parameters of the mechanistic ODE model~(see Fig~\ref{Fig:Flow_India_covid}, Table~\ref{tab:mod1} and method section) on the effect of lockdown and the basic reproduction number ($R_{0}$). Lockdown effect is measured in terms of the differences in total number of COVID-19 cases occurred during May 3, 2020 till May 20, 2020 in \textbf{India} under the lockdown scenarios (\textbf{A1}) and (\textbf{A5}), respectively~(see method section). Effect of Uncertainty of these four parameters on the two mentioned responses are measured using Partial Rank Correlation Coefficients (PRCC). $500$~samples for each parameters were drawn using Latin hypercube sampling techniques (LHS) from their respective ranges provided in Table~\ref{tab:mod1}.}     
	\label{Fig:Sensitivity-IND}
\end{figure}

\clearpage

\begin{center}
	\section*{\Large{\underline{Tables}}}
\end{center}

\begin{table}[ht]
	\captionsetup{font=normalsize}
	\captionsetup{width=1.1\textwidth}
	\tabcolsep 6pt
	\begin{center}
		\caption{Parameters with their respective epidemiological information for the mechanistic ODE model~(see Fig~\ref{Fig:Flow_India_covid} and supplementary) for COVID-19.}
		\label{tab:mod1}
		\begin{tabular}{|c| p{8.2cm} p{3.4cm} p{2cm}}
			\hline\\
			\textbf{Parameters} & \textbf{Biological Meaning} &  \textbf{Value/Ranges} & \textbf{Reference}\\\hline\\
			$\Pi$ $=$ $\mu$ $\times$ $N(0)$ & Recruitment rate of human
			population & Differs over states & - \\\\
			$\displaystyle\frac{1}{\mu}$ & Average life expectancy at birth & Differs over states & \cite{nitiaayog20}\\\\
			$\displaystyle \beta_1$ & Average transmission rate of a symptomatic and asymptomatic COVID-19 infected & (0 - 500) $day^{-1}$  & Estimated \\\\
			$\displaystyle \beta_2$ & Average transmission rate of a notified \& hospitalized COVID-19 infected & (0 - 500) $day^{-1}$  & Estimated \\\\		
			$\displaystyle \rho$ & Fraction of the susceptible population that are exposed to hospital-based transmission & 0 - 0.2 & Estimated \\\
			$\displaystyle \frac{1}{\sigma}$ & COVID-19 incubation period  & (1 - 14) days &  Estimated \\\\
			$\displaystyle \kappa$ & Fraction of the COVID-19 exposed population that become symptomatic infected & 0 - 1 & Estimated\\\\
			$\displaystyle \gamma_1$ & Average recovery rate of symptomatic
			infection & ($\displaystyle \gamma_{3}$ - 1) $day^{-1}$& Estimated\\\\
			$\displaystyle \gamma_2$ &  Average recovery rate of asymptomatic
			infection & ($\displaystyle \gamma_{3}$ - 1) $day^{-1}$ & Estimated \\\\
			$\displaystyle \tau$ & Average hospitalization rate for the COVID-19 symptomatic individuals &  (0 - 1) $day^{-1}$ & Estimated\\\\
			$\displaystyle \delta$ & Average death rate due to COVID-19 infection in hospital & Differs over states & \cite{indiacovid2020track}\\\\
			$\displaystyle \gamma_3$ & Average recovery rate of the notified \& hospitalized populations & Differs over states & \cite{indiacovid2020track}\\\\
			$\displaystyle l$ & Average lockdown rate &(0 - 0.9) $day^{-1}$ & Estimated\\\\
			$\displaystyle \frac{1}{\omega}$ & Current lockdown period in India & 40 days & \cite{timesofindia20b, financialexpress20}\\\hline			
		\end{tabular}
	\end{center}
\end{table}
\clearpage

\begin{table}[ht]
	\captionsetup{font=normalsize}
	\captionsetup{width=1.1\textwidth}
	\tabcolsep 3.5pt
	\footnotesize	
	\centering
	\caption{Estimated uninformative parameters of the mechanistic ODE model combinations~~(see Fig~\ref{Fig:Flow_India_covid} and supplementary). Respective subscripts are \textbf{MH}:~Maharashtra, \textbf{DL}:~Delhi, \textbf{MP}:~Madhya Pradesh, \textbf{RJ}:~Rajasthan, \textbf{GJ}:~Gujarat, \textbf{UP}:~Uttar Pradesh, and \textbf{IND}:~India. All data are given in the format~\textbf{Estimate (95\% CI)}.}\vspace{0.2cm}  
	\begin{tabular}{|p{1.4cm}|ccccccccc} \hline\\
		\textbf{Region} &  $\boldsymbol{\beta_{1}}$ & $\boldsymbol{\beta_{2}}$  & $\boldsymbol{\rho}$ $\times$ $100$  & $\boldsymbol{\sigma}$ & $\boldsymbol{\kappa}$ &  $\boldsymbol{\gamma_{1}}$  & $\boldsymbol{\tau\times 100}$ & $\boldsymbol{\gamma_{2}}$ & $\boldsymbol{l} \times 100$\\ \hline\\
		\textbf{MH}& \footnotesize{$\substack{2.03  \\  \\ (1.05 - 2.75)}$} & \footnotesize{$\substack{21.6  \\  \\ (2.6 - 54.7)}$} & \footnotesize{$\substack{7.98  \\  \\ (0.24-9.78)}$} & \footnotesize{$\substack{0.163  \\  \\ (0.073 - 0.196)}$} & \footnotesize{$\substack{0.118  \\  \\ (0.03 - 0.165)}$} & \footnotesize{$\substack{0.963  \\  \\ (0.144 - 0.974)}$} & \footnotesize{$\substack{5.05  \\  \\ (2.76 - 85.3)}$} & \footnotesize{$\substack{0.863  \\  \\ (0.341 - 0.991)}$}& \footnotesize{$\substack{61.5  \\  \\ (2.51 - 97.3)}$}\\\\
		\hline\\
		\textbf{DL}& \footnotesize{$\substack{1.83 \\  \\ (0.557-3.54)}$} & \footnotesize{$\substack{220.2  \\  \\ (12.97 - 488.8)}$} & \footnotesize{$\substack{7.66 \\  \\ (0.25-9.67)}$} & \footnotesize{$\substack{0.211  \\  \\ (0.075 - 0.748)}$} & \footnotesize{$\substack{0.039  \\  \\ (0.002 - 0.039)}$} & \footnotesize{$\substack{0.96  \\  \\ (0.321 - 0.987)}$} & \footnotesize{$\substack{8.5  \\  \\ (3.8 - 96.3)}$} & \footnotesize{$\substack{0.764  \\  \\ (0.328 - 0.988)}$}&\footnotesize{$\substack{70.8 \\  \\ (13.1 - 87.8)}$}\\\\
		\hline\\
		\textbf{MP}& \footnotesize{$\substack{1.45  \\  \\ (1.04 - 2.91)}$} & \footnotesize{$\substack{42.1  \\  \\ (25.4 - 486)}$} & \footnotesize{$\substack{5.88  \\  \\ (0.24 - 9.72)}$} & \footnotesize{$\substack{0.194  \\  \\ (0.073 - 0.242)}$} & \footnotesize{$\substack{0.014  \\  \\ (0.002 - 0.098)}$} & \footnotesize{$\substack{0.827  \\  \\ (0.141 - 0.980)}$} & \footnotesize{$\substack{14  \\  \\ (3.52 - 92.2)}$} & \footnotesize{$\substack{0.585  \\  \\ (0.21 - 0.978)}$}&\footnotesize{$\substack{34.8  \\  \\ (12.7 - 88.6)}$}\\\\
		\hline\\
		\textbf{RJ}& \footnotesize{$\substack{1.16  \\  \\ (0.557 - 2.74)}$} & \footnotesize{$\substack{325.6  \\  \\ (3.71 - 484)}$} & \footnotesize{$\substack{1.33  \\  \\ (0.07 - 9.56)}$} & \footnotesize{$\substack{0.937  \\  \\ (0.148 - 0.949)}$} & \footnotesize{$\substack{0.0011  \\  \\ (0.001 - 0.003)}$} & \footnotesize{$\substack{0.337 \\  \\ (0.253 - 0.988)}$} & \footnotesize{$\substack{23.2  \\  \\ (4.66 - 96.3)}$} & \footnotesize{$\substack{0.814  \\  \\ (0.244 - 0.978)}$}&\footnotesize{$\substack{21  \\  \\ (10.6 - 81.2)}$}\\\\
		\hline\\
		\textbf{GJ}& \footnotesize{$\substack{0.90  \\  \\ (0.53 - 2.10)}$} & \footnotesize{$\substack{186.3  \\  \\ (9.64 - 490)}$} & \footnotesize{$\substack{2.32  \\  \\ (0.21-9.20)}$} & \footnotesize{$\substack{0.54  \\  \\ (0.101 - 0.94)}$} & \footnotesize{$\substack{0.028  \\  \\ (0.01 - 0.092)}$} & \footnotesize{$\substack{0.84  \\  \\ (0.15 - 0.98)}$} & \footnotesize{$\substack{14  \\  \\ (4.02 - 76.2)}$} & \footnotesize{$\substack{0.52  \\  \\ (0.202 - 0.98)}$}&\footnotesize{$\substack{11.4 \\  \\ (2.31 - 32.2)}$}\\\\
		\hline\\
		\textbf{UP}& \footnotesize{$\substack{0.53  \\  \\ (0.52 - 1.92)}$} & \footnotesize{$\substack{242  \\  \\ (8.23 - 353)}$} & \footnotesize{$\substack{4.29  \\  \\ (0.1 - 9.51)}$} & \footnotesize{$\substack{0.66  \\  \\ (0.078 - 0.71)}$} & \footnotesize{$\substack{0.12  \\  \\ (0.077 - 0.95)}$} & \footnotesize{$\substack{0.78  \\  \\ (0.19 - 0.988)}$} & \footnotesize{$\substack{2.38  \\  \\ (0.2 - 5)}$} & \footnotesize{$\substack{0.408  \\  \\ (0.218 - 0.982)}$}&\footnotesize{$\substack{48.9  \\  \\ (12.3 - 87.6)}$}\\\\
		\hline\\
		\textbf{IND}& \footnotesize{$\substack{2.25  \\  \\ (1.29 - 3.24)}$} & \footnotesize{$\substack{499  \\  \\ (494 - 499.8)}$} & \footnotesize{$\substack{8.99  \\  \\ (0.32-9.70)}$} & \footnotesize{$\substack{0.098  \\  \\ (0.072 - 0.15)}$} & \footnotesize{$\substack{0.35  \\  \\ (0.07 - 0.96)}$} & \footnotesize{$\substack{0.72  \\  \\ (0.38 - 0.99)}$} & \footnotesize{$\substack{0.23  \\  \\ (0.03 - 0.53)}$} & \footnotesize{$\substack{0.995 \\  \\ (0.25 - 0.9959)}$}&\footnotesize{$\substack{76.6  \\  \\ (35.5 - 89.2)}$}\\\\
		\hline
	\end{tabular}
	\label{Tab:estimated-parameters-Table}
\end{table}

\begin{table}[ht]
	\captionsetup{font=normalsize}
	\captionsetup{width=1.1\textwidth}
	\tabcolsep 18pt	
	\centering
	\caption{Estimates of the basic, the community and the hospital reproduction numbers. The informative and uninformative parameters of the mechanistic ODE model (see Fig~\ref{Fig:Flow_India_covid} and supplementary) during the estimation of the different reproduction numbers are taken from Table~\ref{tab:mod1} and Table~\ref{Tab:estimated-parameters-Table}, respectively. Respective subscripts \textbf{MH}, \textbf{DL}, \textbf{MP}, \textbf{RJ}, \textbf{GJ}, \textbf{UP}, and \textbf{IND} are same as Table~\ref{Tab:estimated-parameters-Table}. All data are given in the format~\textbf{Estimate (95\% CI)}.}\vspace{0.3cm}
\label{Tab:reproduction-numbers}	
	\begin{tabular}{|c|ccccc} \hline\\	
		
		\textbf{Region} & $\boldsymbol{R_0}$& $\boldsymbol{R_{C}}$ & \% $\boldsymbol{R_0}$ & $\boldsymbol{R_{H}}$ & \% $\boldsymbol{R_0}$\\\hline
		\\
		\textbf{MH} & $\substack{2.371  \\\\ (2.095 - 4.005)}$ & $\substack{2.309  \\\\ (1.542 - 3.894)}$ & $\substack{97.37  \\\\ (55.46 - 99.93)}$ & $\substack{0.0623  \\\\ (0.0019 - 1.254)}$ & $\substack{2.63  \\\\ (0.1 - 44.54)}$\\\\
		\hline\\
		\textbf{DL} & $\substack{2.54 \\\\ (1.37 - 5.52)}$ & $\substack{2.37  \\\\ (0.934 - 5.49)}$ & $\substack{93.46  \\\\ (56.91 - 99.95)}$ & $\substack{0.1658  \\\\ (0.0016 - 0.8236)}$ & $\substack{6.53  \\\\ (0.048 - 43.09)}$\\\\
		\hline\\
		\textbf{MP}& $\substack{2.497 \\\\ (2.21 - 6.11)}$ & $\substack{2.467  \\\\ (1.68 - 5.63)}$ & $\substack{98.82  \\\\ (49.22 - 99.81)}$ & $\substack{0.0296  \\\\ (0.0062 - 2.045)}$ & $\substack{1.18  \\\\ (0.19 - 50.78)}$ \\\\
		\hline\\
		\textbf{RJ}& $\substack{1.43 \\\\ (1.42 - 3.78)}$ & $\substack{1.42  \\\\ (1.405 - 3.76)}$ & $\substack{99.46  \\\\ (94.40 - 100)}$ & $\substack{0.01 \\\\ (0 - 0.068)}$ & $\substack{0.54  \\\\ (0 - 3.66)}$ \\\\
		\hline\\
		\textbf{GJ}& $\substack{1.835 \\\\ (1.51 - 4.86)}$ & $\substack{1.706  \\\\ (0.895 - 3.72)}$ & $\substack{93.02  \\\\ (31.4 - 99.56)}$ & $\substack{0.128 \\\\ (0.0083 - 2.68)}$ & $\substack{6.98 \\\\ (0.43 - 68.60)}$ \\\\
		\hline\\
		\textbf{UP}& $\substack{1.46 \\\\ (1.35 - 3.14)}$ & $\substack{1.22  \\\\ (0.787 - 2.88)}$ & $\substack{84  \\\\ (39.21 - 99.96)}$ & $\substack{0.233  \\\\ (0.001 - 1.47)}$ & $\substack{16  \\\\ (0.039 - 60.79)}$ \\\\
		\hline\\
		\textbf{IND}& $\substack{2.81 \\\\ (2.15 - 4.94)}$ & $\substack{2.56  \\\\ (1.92 - 4.86)}$ & $\substack{91.28  \\\\ (84.63 - 99.84)}$ & $\substack{0.245 \\\\ (0.0048 - 0.388)}$ & $\substack{8.72  \\\\ (0.16 - 15.37)}$ \\\\
		\hline		
	\end{tabular}
	
	\label{Tab:estimated-R0-Table}
\end{table}
\clearpage
\begin{table}[ht]
	\captionsetup{font=normalsize}
	\captionsetup{width=1.1\textwidth}
	\tabcolsep 5.5pt
	\footnotesize	
	\caption{Ensemble model~(\ref{Eq:ensemble}) forecast of the total notified COVID-19 cases during May 3, 2020 till May 20, 2020. Respective subscripts \textbf{MH}, \textbf{DL}, \textbf{MP}, \textbf{RJ}, \textbf{GJ}, \textbf{UP}, and \textbf{IND} are same as Table~\ref{Tab:estimated-parameters-Table}. Regions where, current lockdown rate~($\boldsymbol{\big\downarrow}$) implies the ensemble model~(\ref{Eq:ensemble}) projections for the Scenarios (A2) to (A4) with 10\%, 20\% and 30\% decrement in the current estimate of lockdown rate (see Table~\ref{Tab:estimated-parameters-Table}) during the mentioned period. Whereas, current lockdown rate~($\big\uparrow$) implies the ensemble model~(\ref{Eq:ensemble}) projections for the Scenarios (A2) to (A4) with 10\%, 20\% and 30\% increment in the current estimate of lockdown rate (see Table~\ref{Tab:estimated-parameters-Table}) during the mentioned period. Scenario (A1) implies the ensemble model~(\ref{Eq:ensemble}) forecast with the current estimate of the lockdown rate (see Table~\ref{Tab:estimated-parameters-Table}) during May 3, 2020 till May 20, 2020. Finally, Scenario (A5) implies the ensemble model~(\ref{Eq:ensemble}) forecast with no lockdown during May 3, 2020 till May 20, 2020. All data are provided in the format~\textbf{Estimate (95\% CI)}.}
	
	\begin{tabular}{|p{1.5cm}|p{3.2cm} c c c c c} \hline\\
		\textbf{Region} & \textbf{Current lockdown rate} &\textbf{Scenario~A1}
		& \textbf{Scenario~A2}
		& \textbf{Scenario~A3} & \textbf{Scenario~A4} & \textbf{Scenario~A5} 
		\\\\ \hline\\
		\textbf{MH}&\contour{red}{$\big\downarrow$}  & \footnotesize{$\substack{16790 \\  \\  (16056 - 18727)}$} & \footnotesize{$\substack{17410  \\  \\  (16499 - 19812)}$} & \footnotesize{$\substack{18147  \\  \\  (17026 - 21101)}$} & \footnotesize{$\substack{19036  \\  \\  (17662 - 22656)}$} & \footnotesize{$\substack{23000  \\  \\  (20499 - 29596)}$}\\\\
		\hline\\
		\textbf{DL}& \contour{red}{$\big\downarrow$}& \footnotesize{$\substack{2642  \\  \\ (1791 - 3294)}$} & \footnotesize{$\substack{1840 \\  \\ (1337 - 3297)}$} & \footnotesize{$\substack{2693 \\  \\ (1901 - 3300)}$} & \footnotesize{$\substack{2729  \\  \\ (1979 - 3304)}$} & \footnotesize{$\substack{2886  \\  \\  (2319 - 3320)}$}\\\\
		\hline\\
		\textbf{MP}& \contour{red}{$\big\uparrow$} & \footnotesize{$\substack{4036 \\  \\ (3599 - 4578)}$} & \footnotesize{$\substack{3924 \\  \\ (3531 - 4410)}$} & \footnotesize{$\substack{3825 \\  \\  (3472 - 4263)}$} & \footnotesize{$\substack{3738 \\  \\ (3419 - 4134)}$} & \footnotesize{$\substack{5058  \\  \\  (4216 - 6101)}$}\\\\
		\hline\\
		\textbf{RJ}& \contour{red}{$\big\uparrow$} & \footnotesize{$\substack{1184 \\  \\(654 - 1624)}$} & \footnotesize{$\substack{1156 \\  \\(611 - 1608)}$} & \footnotesize{$\substack{1133\\  \\ (576 - 1595)}$} & \footnotesize{$\substack{1114 \\  \\(547 - 1585)}$} & \footnotesize{$\substack{1318 \\  \\(858 - 1699)}$}\\\\
		\hline\\
		\textbf{GJ}& \contour{red}{$\big\uparrow$} & \footnotesize{$\substack{7707 \\  \\(7689 - 12589)}$} & \footnotesize{$\substack{7444 \\  \\(7427 - 11991)}$} & \footnotesize{$\substack{7206\\  \\ (7191 - 11451)}$} & \footnotesize{$\substack{6990 \\  \\(6976 - 10961)}$} & \footnotesize{$\substack{9337 \\  \\(9312 - 16292)}$}\\\\
		\hline\\
		\textbf{UP}& \contour{red}{$\big\uparrow$} & \footnotesize{$\substack{5599 \\  \\(3656 - 6478)}$} & \footnotesize{$\substack{5383 \\  \\(3564 - 6206)}$} & \footnotesize{$\substack{5184\\  \\ (3479 - 5954)}$} & \footnotesize{$\substack{4999 \\  \\(3401 - 5722)}$} & \footnotesize{$\substack{7217 \\  \\(4343 - 8517)}$}\\\\
		\hline\\
		\textbf{IND}& \contour{red}{$\big\downarrow$} & \footnotesize{$\substack{38134 \\  \\(36550 - 45296)}$} & \footnotesize{$\substack{40278 \\  \\(38877 - 46612)}$} & \footnotesize{$\substack{42824\\  \\ (41640 - 48174)}$} & \footnotesize{$\substack{45896 \\  \\(44975 - 50059)}$} & \footnotesize{$\substack{57159 \\  \\(56971 - 57201)}$}\\\\
		\hline
	\end{tabular}
	\label{Tab:cases-preiction-Table}
\end{table}
\clearpage
\setcounter{equation}{0}
\setcounter{figure}{0}  
\setcounter{table}{0} 
\renewcommand{\theequation}{S-\arabic{equation}}
\renewcommand{\thetable}{S\arabic{table}}
\renewcommand{\thelemma}{S\arabic{lemma}}
\renewcommand{\thetheorem}{S\arabic{theorem}}
\renewcommand{\theproposition}{S\arabic{proposition}}
\renewcommand{\thefigure}{S\arabic{figure}}

\begin{center}
	\section*{\Large{Supplementary Method}}	
\end{center}

\subsection*{\textbf{Model without lock-down}}
The basic structure of our model is an extension of our earlier model~\cite{sardar2020assessment} by considering hospital-based transmission. As hospitalized \& notified infected ($H$) can only contact with a very small fraction of the susceptible class therefore, we assume the effective contact for a single hospitalized \& notified infected is $\displaystyle \beta_{2} \frac{(\rho S)}{N}$ and $\rho$ is the fraction of susceptible that are in contact with the hospitalized \& notified infected. We also assume that effective contact of a single symptomatic ($I_{S}$) and a asymptomatic individual is same and at a rate $\displaystyle \beta_{1} \frac{ S}{N}$. Therefore, total human population $N(t)$ at time $t$ is subdivided into six mutually exclusive sub-classes namely, susceptible ($S(t)$), exposed ($E(t)$), symptomatic ($I_{S}(t)$), asymptomatic ($I_{A}(t)$), hospitalized \& notified ($H(t)$), and recovered ($R(t)$), respectively. Interaction between different sub-classes are provided by the following system of differential equation:

\begin{eqnarray}\label{EQ:eqn 2.1}
\displaystyle{ \frac{dS}{dt} } &=& \Pi-\frac{\beta_1 S I_S}{N}-\frac{\beta_1 S I_A}{N}-\frac{\beta_2  (\rho S) H}{N}-\mu S\nonumber\\
\displaystyle{ \frac{dE}{dt} } &=& \frac{\beta_1 S I_S}{N}+\frac{\beta_1 S I_A}{N}+\frac{\beta_2  (\rho S) H}{N}-(\mu+\sigma)E,\nonumber \\
\displaystyle{ \frac{dI_{S}}{dt} } &=& \kappa \sigma E-(\tau+\mu+\gamma_1)I_S ,\nonumber \\
\displaystyle{ \frac{dI_A}{dt} } &=&  (1-\kappa)\sigma E-(\mu+\gamma_2)I_A, \\
\displaystyle{\frac{dH}{dt} } &=& \tau I_S-(\gamma_3+\delta+\mu)H, \nonumber \\
\displaystyle{ \frac{dR}{dt} } &=& \gamma_1 I_S+\gamma_2 I_A+\gamma_3 H-\mu R.\nonumber
\end{eqnarray}  

\subsection*{\textbf{Model with lock-down}}
Lock-down effect in the model~(\ref{EQ:eqn 2.1}) can be included by home quarantine a fraction of the susceptible population at a rate $l$. This home-quarantine populations are considered as a new subclass namely, lock-down class $L$. We assume that population in the lock-down compartment are home quarantined and they do not interact with other populations. Therefore, effective contact of a single symptomatic ($I_{S}$), asymptomatic ($I_{A}$), and hospitalized \& notified infected are $\displaystyle \beta_{1} \frac{ S}{(N-L)}$, $\displaystyle \beta_{1} \frac{ S}{(N-L)}$, and $\displaystyle \beta_{1} \frac{ (\rho S)}{(N-L)}$, respectively. The interaction of all these sub-classes can be expressed as an lock-down ODE model as follows:    

\begin{eqnarray}\label{EQ:eqn 2.2}
\displaystyle{ \frac{dS}{dt} } &=& \Pi+\omega L-\frac{\beta_1 S I_S}{(N-L)}-\frac{\beta_1 S I_A}{(N-L)}-\frac{ \beta_2  (\rho S) H}{(N-L)}-(\mu+l) S\nonumber\\
\displaystyle{ \frac{dL}{dt} } &=& l S-(\mu+\omega)L,\nonumber \\
\displaystyle{ \frac{dE}{dt} } &=& \frac{\beta_1 S I_S}{(N-L)}+\frac{\beta_1 S I_A}{(N-L)}+\frac{ \beta_2  (\rho S) H}{(N-L)}-(\mu+\sigma)E,\nonumber \\
\displaystyle{ \frac{dI_s}{dt} } &=& \kappa \sigma E-(\tau+\mu+\gamma_1)I_S ,\nonumber \\
\displaystyle{ \frac{dI_A}{dt} } &=&  (1-\kappa)\sigma E-(\mu+\gamma_2)I_A, \\
\displaystyle{\frac{dH}{dt} } &=& \tau I_S-(\delta+\mu+\gamma_3)H, \nonumber \\
\displaystyle{ \frac{dR}{dt} } &=& \gamma_1 I_S+\gamma_2 I_A+\gamma_3 H-\mu R.\nonumber
\end{eqnarray} 

A flow diagram of the model~(\ref{EQ:eqn 2.1}) and~(\ref{EQ:eqn 2.2})  are provided in \textbf{Fig~1} (see main text). Biological interpretations of the model~(\ref{EQ:eqn 2.1}) and~(\ref{EQ:eqn 2.2}) parameters are provided in \textbf{Table~1} (see main text).

\subsection*{Positivity and boundedness of the solution for the Model \eqref{EQ:eqn 2.1}}
In this section, we provided a proof for the positivity and boundedness of solutions of the system \eqref{EQ:eqn 2.1} with initial conditions $(S(0),E(0),I_S(0),I_A(0),H(0),R(0))^T\in \mathbb{R}_{+}^6$. We first state the following lemma.
\begin{lemma}(\cite{yang1996permanence}) 
	Suppose $\Omega \subset \mathbb{R} \times \mathbb{C}^n$ is open, $f_i \in C(\Omega, \mathbb{R}), i=1,2,3,...,n$. If $f_i|_{x_i(t)=0,X_t \in \mathbb{C}_{+0}^n}\geq 0$, $X_t=(x_{1t},x_{2t},.....,x_{1n})^T, i=1,2,3,....,n$, then $\mathbb{C}_{+0}^n\lbrace \phi=(\phi_1,.....,\phi_n):\phi \in \mathbb{C}([-\tau,0],\mathbb{R}_{+0}^n)\rbrace$ is the invariant domain of the following equations
	\begin{align*}
	\frac{dx_i(t)}{dt}=f_i(t,X_t), t\geq \sigma, i=1,2,3,...,n.
	\end{align*}
	where $\mathbb{R}_{+0}^n=\lbrace (x_1,....x_n): x_i\geq 0, i=1,....,n \rbrace$.
	\label{lma1}	
\end{lemma}
\begin{proposition}
	The system \eqref{EQ:eqn 2.1} is invariant in $\mathbb{R}_{+}^6$.
\end{proposition}
\begin{proof}
	By re-writing the system \eqref{EQ:eqn 2.1}, we have
	\begin{eqnarray}
	\frac{dX}{dt} & =B(X(t)), X(0)=X_0\geq 0
	\label{EQ:eqn A.1}
	\end{eqnarray}
	$ B(X(t))=(B_1(X),B_1(X),...,B_6(X))^T$\\
	We note that
	\begin{align*}
	\frac{dS}{dt}|_{S=0}=\Pi \geq 0,	
	\frac{dE}{dt}|_{E=0}=\frac{(\beta_1 I_S+\beta_1 I_A+\rho \beta_2 H)S}{N}\geq 0,\\
	\frac{dI_S}{dt}|_{I_S=0}=\kappa \sigma E \geq 0,
	\frac{dI_A}{dt}|_{I_A=0}=(1-\kappa)\sigma E\geq 0,\\ \frac{dH}{dt}|_{H=0}=\tau I_S \geq 0,
	\frac{dR}{dt}|_{R=0}=\gamma_1 I_S+\gamma_2 I_A+\gamma_3 H \geq 0.
	\end{align*}
	Then it follows from the Lemma \ref{lma1} that $\mathbb{R}_{+}^6$ is an invariant set.
\end{proof}

\begin{lemma}
	The system \eqref{EQ:eqn 2.1} is bounded in the region\\ $\Omega=\lbrace(S,E,I_S,I_A,H,R\in \mathbb{R}_+^{6}|S+E+I_S+I_A+H+R\leq \frac{\Pi}{\mu}\rbrace$
\end{lemma}
\begin{proof}
	We observed from the system that
	\begin{align*}
	&\frac{dN}{dt}=\Pi-\mu N-\delta H\leq \Pi-\mu N\\
	& \Longrightarrow \lim\limits_{t\rightarrow \infty}sup N(t)\leq \frac{\Pi}{\mu}
	\end{align*}
	Hence the system \eqref{EQ:eqn 2.1} is bounded.
\end{proof}
\subsection*{Local stability of disease-free equilibrium (DFE)}
The DFE of the model \eqref{EQ:eqn 2.1} is given by
\begin{align*}
\varepsilon_0&=(S^0, E^0, I_S^0, I_A^0, H^0, R^0)\\
&=\Big(\frac{\Pi }{\mu}, 0, 0, 0, 0, 0, 0\Big)
\end{align*}
The local stability of $\varepsilon_0$ can be established on the system \eqref{EQ:eqn 2.1} by using the next generation operator method. Using the notation in \cite{van2002reproduction}, the matrices $F$ for the new infection and $V$ for the transition terms are given, respectively, by
\begin{align*}
F&=\begin{bmatrix}
0 & \beta_1 &  \beta_1 & \rho \beta_2 \\
0&0 & 0 &0  \\
0 &0 & 0 &0  \\
0 &0&0 &0  \\
\end{bmatrix},\\
V&=\begin{bmatrix}
\mu+\sigma &0 & 0 &0 \\
-\kappa \sigma &\tau+\mu+\gamma_1 & 0 &0 \\
-(1-\kappa) \sigma &0 & \mu + \gamma_2 &0 \\
0 &-\tau & 0  &\delta+\mu+\gamma_3\\
\end{bmatrix}.
\end{align*}
It follows that the basic reproduction number \cite{hethcote2000mathematics}, denoted by $R_0=\rho(FV^{-1})$, where $\rho$ is the spectral radius, is given by
\begin{align*}
R_0=\frac{\beta_1 \kappa \sigma}{(\mu + \sigma)(\gamma_1 + \mu + \tau)}+\frac{\beta_1 (1-\kappa) \sigma}{(\mu+\gamma_2)(\mu + \sigma)} +\frac{\beta_2 \kappa \rho \sigma \tau}{(\mu + \sigma)(\delta + \gamma_3 + \mu)(\gamma_1 + \mu + \tau)}
\end{align*}
Using Theorem 2 in \cite{van2002reproduction}, the following result is established.
\begin{lemma}	\label{EQ:eqn L.1}
	The DFE, $\varepsilon_0$, of the model \eqref{EQ:eqn 2.1} is locally-asymptotically stable (LAS) if $R_0<1$, and unstable if $R_0>1$.
\end{lemma}

\subsection*{Global stability of DFE}
To show the global stability of $\varepsilon_0$, of the model \eqref{EQ:eqn 2.1}, we rewrite the system~\eqref{EQ:eqn 2.1} as follows 

\begin{eqnarray}\label{EQ:eqn 2.3}
\displaystyle{ \frac{dX}{dt} } &=&  T(X,I)\\
\displaystyle{ \frac{dI}{dt} } &=& G(X,I), \hspace{0.3cm} G(X,0) = 0,\nonumber
\end{eqnarray}

where, $X = (S,R) \in \mathbb{R}_+^{2}$ denotes (it components) the number of uninfected individuals and $I = (E, I_{S}, I_{A}, H) \in \mathbb{R}_+^{4}$ denotes (it components) the number of infected individuals including latent, infectious, etc. $\varepsilon_0 =(X^{*}, 0)$ denotes the disease-free equilibrium of the system~(\ref{EQ:eqn 2.3}). For the system~(\ref{EQ:eqn 2.1}), $T(X,I)$ and $G(X,I)$ are given as follows: 

\begin{align*}
T (X,I) &=\begin{bmatrix}
\Pi-\frac{\beta_1 S I_S}{N}-\frac{\beta_1 S I_A}{N}-\frac{\beta_2  (\rho S) H}{N}-\mu S \\\\
\gamma_1 I_S+\gamma_2 I_A+\gamma_3 H-\mu R\\
\end{bmatrix},\\\\
G(X,I)&=\begin{bmatrix}
\frac{\beta_1 S I_S}{N}+\frac{\beta_1 S I_A}{N}+\frac{\beta_2  (\rho S) H}{N}-(\mu+\sigma)E\\\\
\kappa \sigma E-(\tau+\mu+\gamma_1)I_S\\\\
(1-\kappa)\sigma E-(\mu+\gamma_2)I_A\\\\
\tau I_S-(\gamma_3+\delta+\mu)H\\
\end{bmatrix}.
\end{align*} It is clear from the expression of $G(X,I)$ that $G(X,0) =0 $. 

To show global stability of $\varepsilon_0 =(X^{*}, 0)$ following two condition must hold:

\textbf{(H1)} For $\displaystyle{ \frac{dX}{dt} } =  T(X,0)$, $X^{*}$ is globally asymptotically stable.\\ 

\textbf{(H2)} $G(X,I) = A I - \hat{G}(X,I)$, $\hat{G}(X,I) \geq 0$ for $(X,I) \in \Omega$,\\

where $A = D_{I}G(X^{*},0)$ is an M-matrix (the off diagonal elements are nonnegative) and $\Omega$ is the region where model~(\ref{EQ:eqn 2.1}) makes biological sense. 

Now, the system defined in \textbf{(H1)} can be written as \begin{eqnarray}
\frac{dS}{dt} & =\Pi -\mu S,\\
\frac{dR}{dt} & = -\mu R.\nonumber
\label{EQ:eqn A.4}
\end{eqnarray}
Solving analytically this system of equation we get, $S(t)=\frac{\Pi}{\mu}+e^{-\mu t}(S(0)-\frac{\Pi}{\mu})$, $R(t)=e^{-\mu t}R(0)$. As $t \rightarrow \infty$, $S(t)=\frac{\Pi}{\mu},~ R(t) \rightarrow 0$. Therefore, $X^{*}$ is globally asymptotically stable for $\displaystyle{ \frac{dX}{dt} } =  T(X,0)$. 

So \textbf{(H1)} holds for the system~(\ref{EQ:eqn 2.1}).
Now matrices $A$ and $\hat{G}(X,I)$ for the system~(\ref{EQ:eqn 2.1}) are given as follows:  

\begin{align*}
A &=\begin{bmatrix}
-(\mu + \sigma) & \beta_{1} & \beta_{1} & \rho \beta_{1} \\\\
\kappa \sigma & - (\gamma_1 + \tau + \mu)& 0 &0 \\\\
(1-\kappa) \sigma & 0 & - (\gamma_2 + \mu) &0 \\\\
0 & \tau & 0 & - (\gamma_3 + \delta + \mu)\\
\end{bmatrix},\\\\
\hat{G}(X,I) &=\begin{bmatrix}
\beta_{1} I_{S} (1- \frac{S}{N}) + \beta_{1} I_{A} (1- \frac{S}{N}) + \rho \beta_{2} H (1- \frac{S}{N})\\\\
0\\\\
0 \\\\
0\\
\end{bmatrix}.
\end{align*}

Clearly, $A$ is an M-matrix and as $S(t) \leq N (t)$ in $\Omega$, therefore, $\hat{G}(X,I) \geq 0$ for $(X,I) \in \Omega$. Following \cite{castillo2002computation}, the below result can be stated:

\begin{theorem}
	The DFE of the model \eqref{EQ:eqn 2.1} is globally asymptotically stable in $\Omega$ whenever $R_0 < 1$.
\end{theorem}

\subsection*{Existence of endemic equilibria}

In this section, the existence of the endemic equilibrium of the model \eqref{EQ:eqn 2.1} is established. Let us denote
\begin{align*}
k_1 & =\mu+\sigma, k_2=\gamma_1+\mu+\tau, k_3=\mu+\gamma_2, k_4=\delta+\gamma_3+\mu.
\end{align*}
Let $\varepsilon^*=(S^*, L^*, E^*, I_S^*, I_A^*, H^*, R^*)$ represents any arbitrary endemic equilibrium point (EEP) of the model \eqref{EQ:eqn 2.1}. Further, define
\begin{align}\label{EQ:eqn A.2}
\lambda^*=\frac{\beta_1 I_S^*}{N^*}+\frac{\beta_1  I_A^*}{N^*}+\frac{\rho \beta_2  H^*}{N^*}
\end{align}
It follows, by solving the equations in \eqref{EQ:eqn 2.1} at steady-state, that
\begin{align}\label{EQ:eqn A.3}
S^*&=\frac{\Pi}{\lambda^*+\mu}, E^*=\frac{\lambda^*S^*}{k_1},  I_S^*=\frac{\kappa\sigma \lambda^*S^*}{k_1k_2}, I_A^* =\frac{(1-\kappa) \sigma \lambda^*S^*}{k_1k_3}\\\nonumber
H^*& =\frac{\tau \kappa \sigma \lambda^*S^*}{k_1 k_2 k_4}, R^*=\frac{\gamma_1 \kappa\sigma \lambda^*S^*}{\mu k_1k_2}+ \frac{\gamma_2 (1-\kappa) \sigma \lambda^*S^*}{\mu k_1k_3}+ \frac{\gamma_3 \tau \kappa \sigma \lambda^*S^*}{\mu k_1 k_2 k_4} \nonumber
\end{align}
Substituting the expression in \eqref{EQ:eqn A.3} into \eqref{EQ:eqn A.2} shows that the non-zero equilibrium of the model \eqref{EQ:eqn 2.1} satisfy the following linear equation, in terms of $\lambda^*$:
\begin{align}
a_0 \lambda^* +a_1=0
\end{align}
where
\begin{align*}
a_0&= \mu k_2 k_3 k_4 + \kappa\sigma k_3 k_4 (\mu+\gamma_1)+ (1-\kappa)\sigma k_2 k_4 (\mu+\gamma_2) +\tau \kappa \sigma k_3(\mu+\gamma_3)\\
a_1&=\mu k_1 k_2 k_3 k_4(1-R_0)
\end{align*}
Since $a_0>0$, $\mu>0$, $k_1>0$, $k_2>0$, $k_3>0$ and $k_4>0$, it is clear that the model \eqref{EQ:eqn 2.1} has a unique endemic equilibrium point (EEP) whenever $R_0>1$ and no positive endemic equilibrium point whenever $R_0<1$. This rules out the possibility of the existence of equilibrium other than DFE whenever $R_0<1$. Furthermore, it can be shown that, the DFE $\varepsilon_0$ of the model \eqref{EQ:eqn 2.1} is globally asymptotically stable (GAS) whenever $R_0<1$.

From the above discussion we have concluded that
\begin{theorem}
	The model \eqref{EQ:eqn 2.1} has a unique endemic (positive) equilibrium, given by $\varepsilon^*$, whenever $R_0>1$ and has no endemic equilibrium for $R_0\leq 1$.
\end{theorem}
\clearpage

\begin{center}
	\section*{\Large{\underline{Figures}}}
\end{center}

\begin{figure}[ht]
	\captionsetup{width=1.1\textwidth}
	{\includegraphics[width=1.15\textwidth]{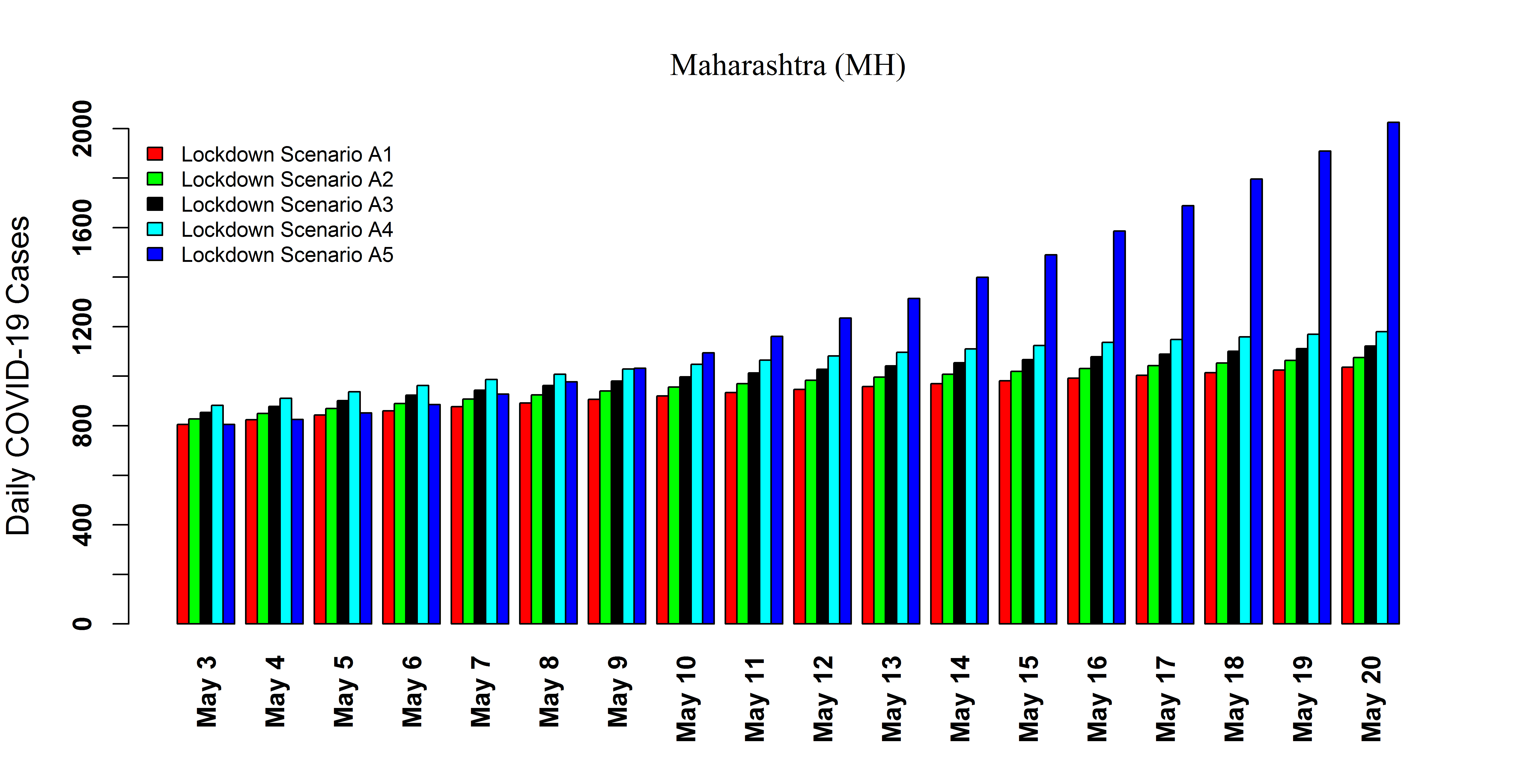}}
	\caption{Ensemble model~(see main text) forecast for the daily notified COVID-19 cases in Maharashtra during May 3, 2020 till May 20, 2020 under five different lockdown scenarios. Under lockdown scenarios (\textbf{A2}) to (\textbf{A4}), projections are made with 10\%, 20\% and 30\% \textbf{decrement} in the current estimate of the lockdown rate, respectively (see Table~2 in main text). In lockdown scenario (\textbf{A1}), projections are made with current estimated lockdown rate (see Table~2 in main text). Finally, in the lockdown scenario (\textbf{A5}), projections are made with no lockdown during May 3, 2020 till May 20, 2020.}  
	\label{Fig:Forecast-MH}
\end{figure}

\begin{figure}[ht]
	\captionsetup{width=1.1\textwidth}
	{\includegraphics[width=1.15\textwidth]{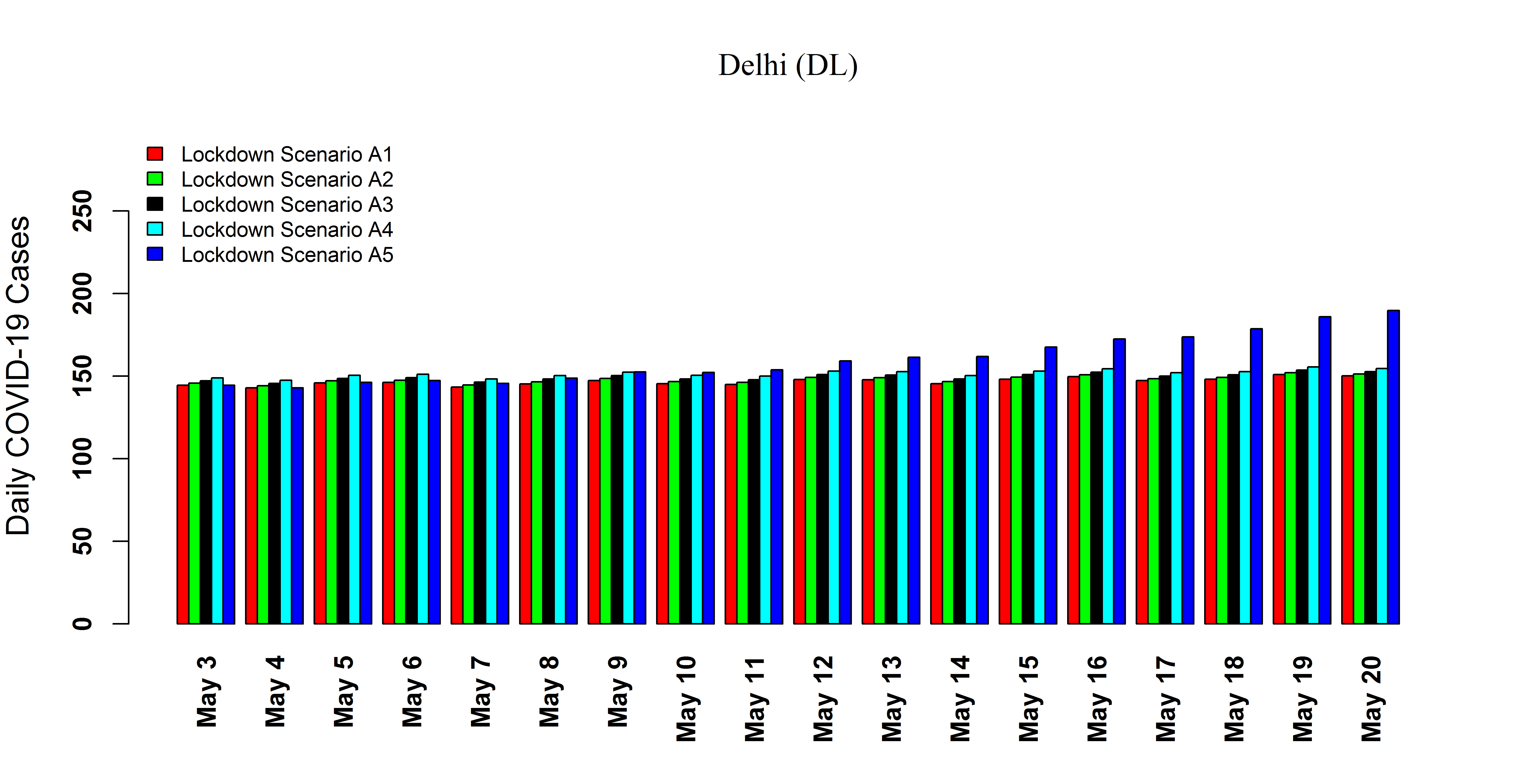}}
	\caption{Ensemble model~(see main text) forecast for the daily notified COVID-19 cases in Delhi during May 3, 2020 till May 20, 2020 under five different lockdown scenarios. Lockdown Scenarios (\textbf{A1}) to~(\textbf{A5}) are same as Figure~\ref{Fig:Forecast-MH}.}     
	\label{Fig:Forecast-DL}
\end{figure}

\begin{figure}[ht]
	\captionsetup{width=1.1\textwidth}
	{\includegraphics[width=1.15\textwidth]{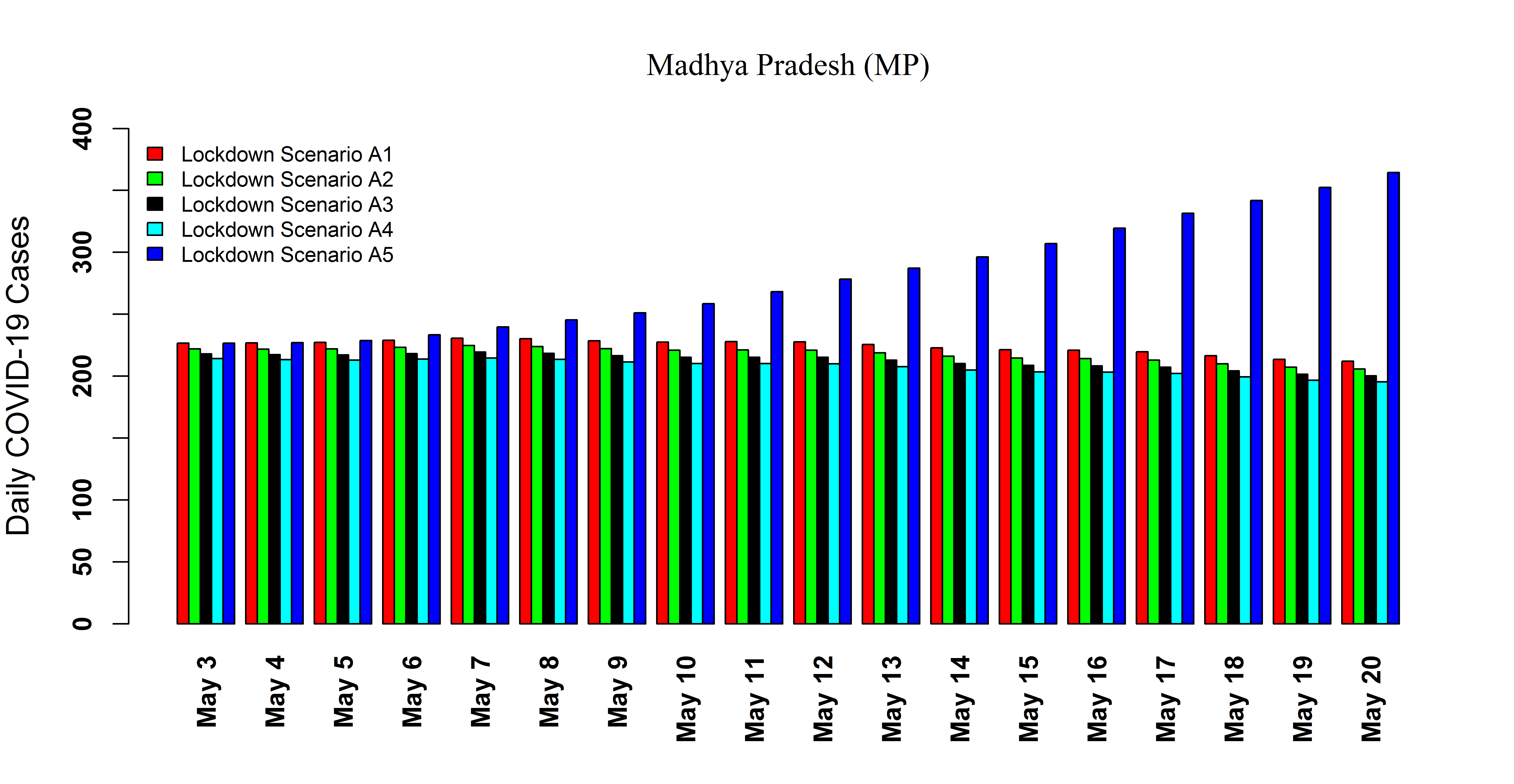}}
	\caption{Ensemble model~(see main text) forecast for the daily notified COVID-19 cases in Madhya Pradesh during May 3, 2020 till May 20, 2020 under five different lockdown scenarios. Lockdown Scenarios (\textbf{A1}) and~(\textbf{A5}) are same as Figure~\ref{Fig:Forecast-MH}. However, under lockdown scenarios (\textbf{A2}) to (\textbf{A4}), projections are made with 10\%, 20\% and 30\% \textbf{increment} in the current estimate of the lockdown rate, respectively (see Table~2 in main text).}     
	\label{Fig:Forecast-MP}
\end{figure}

\begin{figure}[ht]
	\captionsetup{width=1.1\textwidth}
	{\includegraphics[width=1.15\textwidth]{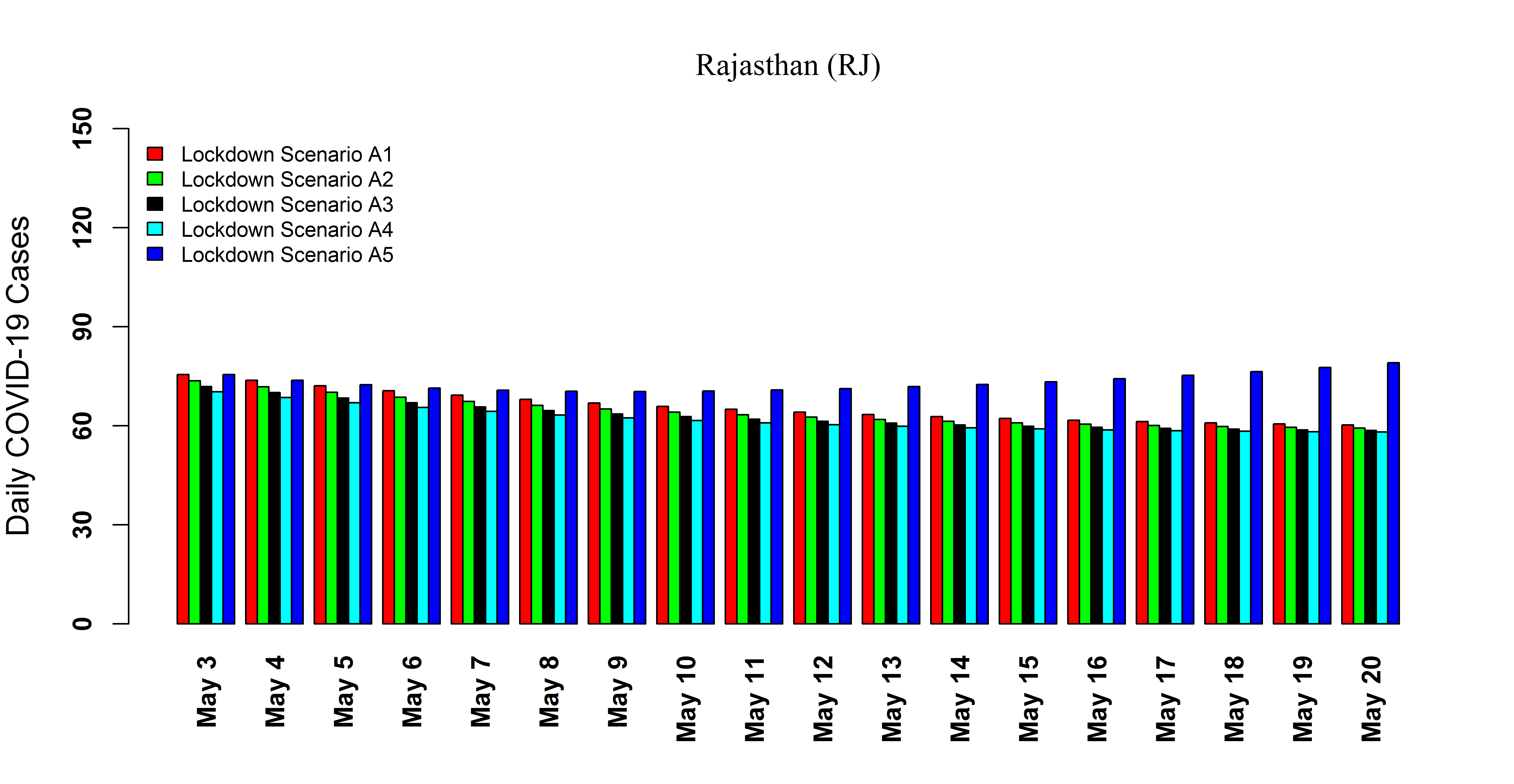}}
	\caption{Ensemble model~(see main text) forecast for the daily notified COVID-19 cases in Rajasthan during May 3, 2020 till May 20, 2020 under five different lockdown scenarios. Lockdown Scenarios (\textbf{A1}) and (\textbf{A5}) are same as Figure~\ref{Fig:Forecast-MH}. However, under lockdown scenarios (\textbf{A2}) to (\textbf{A4}), projections are made with 10\%, 20\% and 30\% \textbf{increment} in the current estimate of the lockdown rate, respectively (see Table~2 in main text).}     
	\label{Fig:Forecast-RJ}
\end{figure}

\begin{figure}[ht]
	\captionsetup{width=1.1\textwidth}
	{\includegraphics[width=1.15\textwidth]{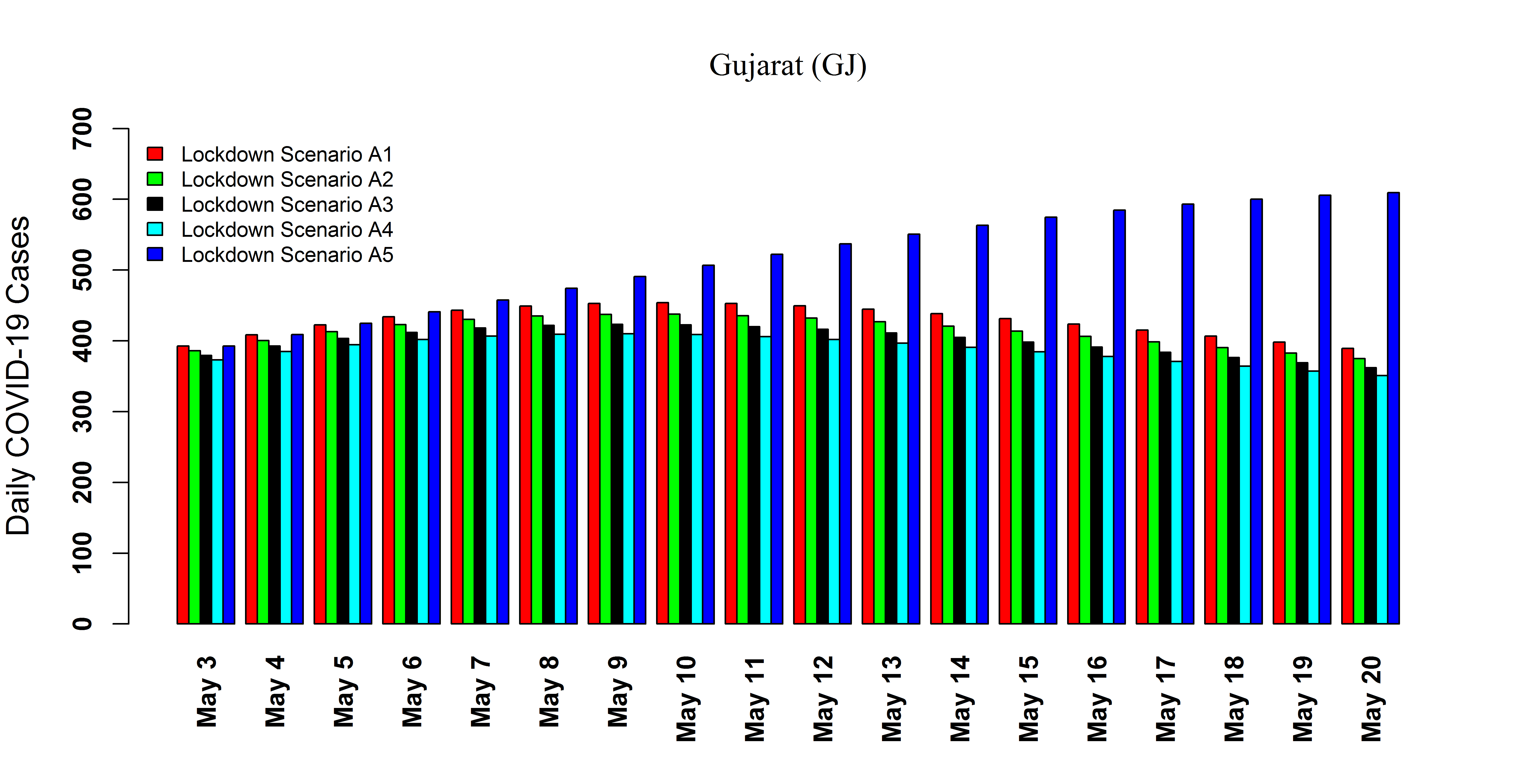}}
	\caption{Ensemble model~(see main text) forecast for the daily notified COVID-19 cases in Gujarat during May 3, 2020 till May 20, 2020 under five different lockdown scenarios. Lockdown Scenarios (\textbf{A1}) and (\textbf{A5}) are same as Figure~\ref{Fig:Forecast-MH}. However, under lockdown scenarios (\textbf{A2}) to (\textbf{A4}), projections are made with 10\%, 20\% and 30\% \textbf{increment} in the current estimate of the lockdown rate, respectively (see Table~2 in main text).}     
	\label{Fig:Forecast-GJ}
\end{figure}

\begin{figure}[ht]
	\captionsetup{width=1.1\textwidth}
	{\includegraphics[width=1.15\textwidth]{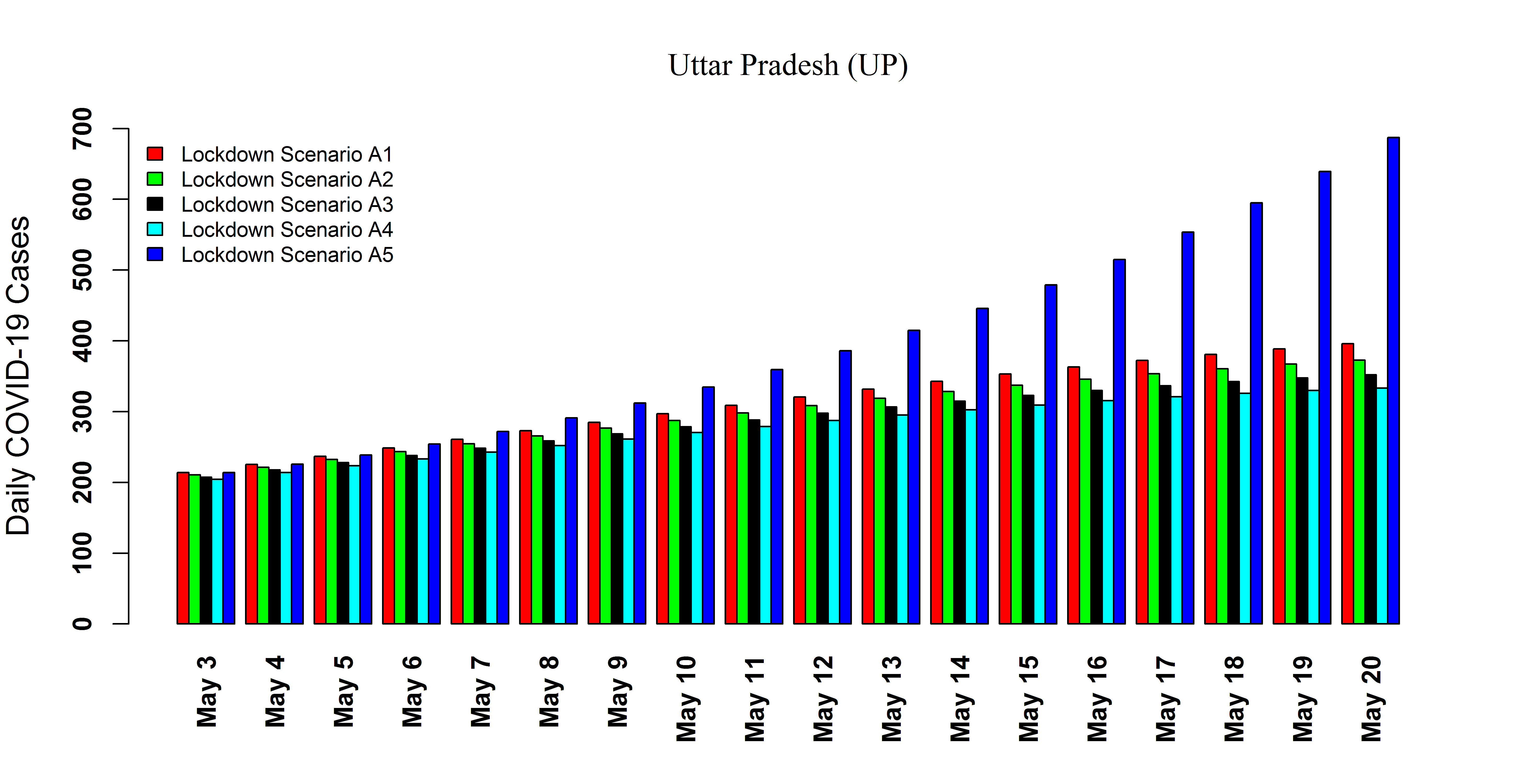}}
	\caption{Ensemble model~(see main text) forecast for the daily notified COVID-19 cases in Uttar Pradesh during May 3, 2020 till May 20, 2020 under five different lockdown scenarios. Lockdown Scenarios (\textbf{A1}) and (\textbf{A5}) are same as Figure~\ref{Fig:Forecast-MH}. However, under lockdown scenarios (\textbf{A2}) to (\textbf{A4}), projections are made with 10\%, 20\% and 30\% \textbf{increment} in the current estimate of the lockdown rate, respectively (see Table~2 in main text).}     
	\label{Fig:Forecast-UP}
\end{figure}

\begin{figure}[ht]
	\captionsetup{width=1.1\textwidth}
	{\includegraphics[width=1.15\textwidth]{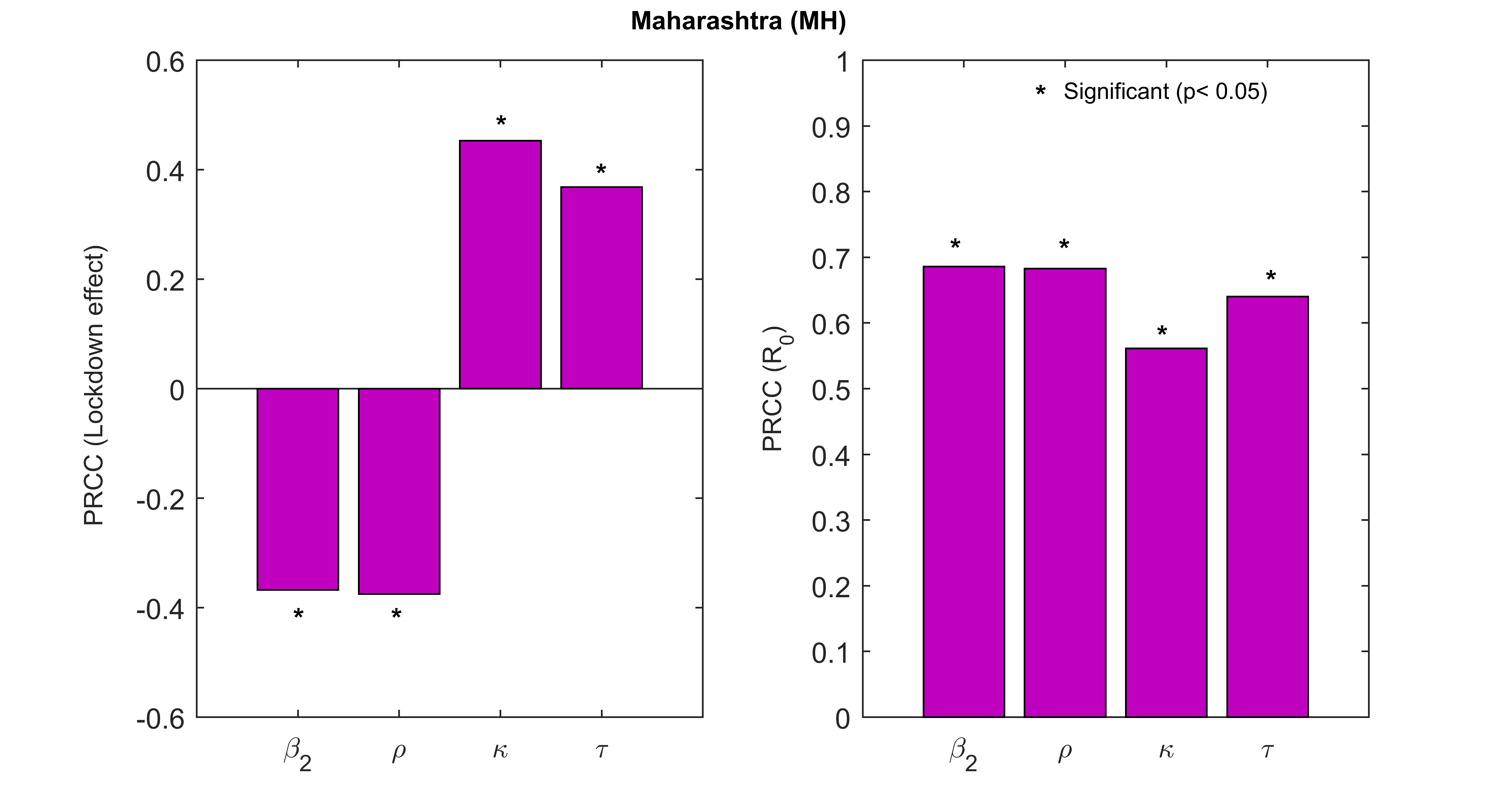}}
	\caption{Effect of uncertainty of four epidemiologically measurable and controllable parameters of the mechanistic ODE model~(see Table~1 and Fig~1 in the main text) on the effect of lockdown and the basic reproduction number ($R_{0}$). Lockdown effect is measured in terms of the differences in total number of COVID-19 cases occurred during May 3, 2020 till May 20, 2020 in \textbf{Maharashtra} under the lockdown scenarios (\textbf{A1}) and (\textbf{A5}), respectively~(see main text). Effect of Uncertainty of these four parameters on the two mentioned responses are measured using Partial Rank Correlation Coefficients (PRCC). $500$~samples for each parameters were drawn using Latin hypercube sampling techniques (LHS) from their respective ranges provided in Table~1 (main text).}     
	\label{Fig:Sensitivity-MH}
\end{figure}

\begin{figure}[ht]
	\captionsetup{width=1.1\textwidth}
	{\includegraphics[width=1.15\textwidth]{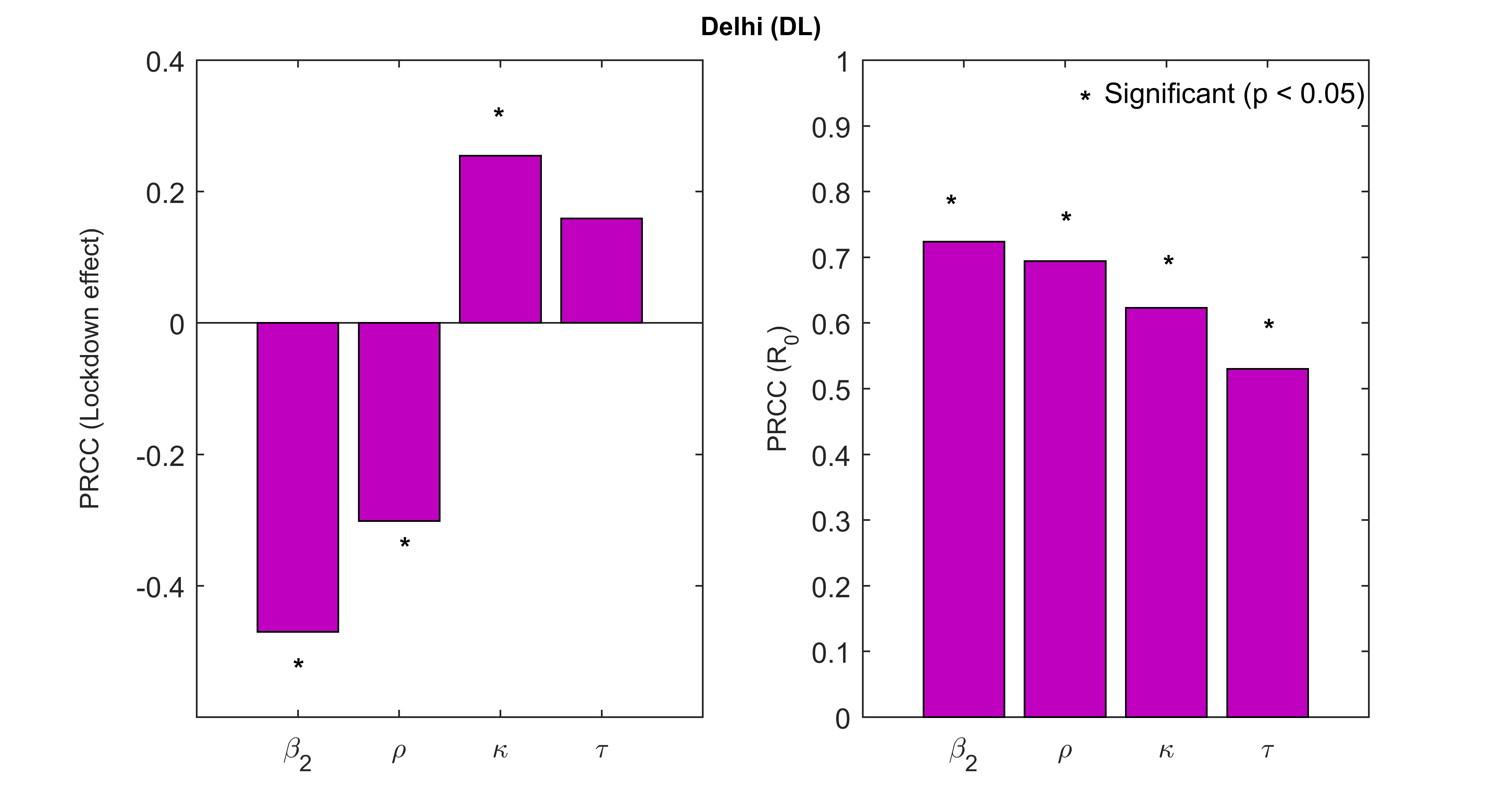}}
	\caption{Effect of uncertainty of four epidemiologically measurable and controllable parameters of the mechanistic ODE model~(see Table~1 and Fig~1 in the main text) on the effect of lockdown and the basic reproduction number ($R_{0}$). Lockdown effect is measured in terms of the differences in total number of COVID-19 cases occurred during May 3, 2020 till May 20, 2020 in \textbf{Delhi} under the lockdown scenarios (\textbf{A1}) and (\textbf{A5}), respectively~(see main text). Effect of Uncertainty of these four parameters on the two mentioned responses are measured using Partial Rank Correlation Coefficients (PRCC). $500$~samples for each parameters were drawn using Latin hypercube sampling techniques (LHS) from their respective ranges provided in Table~1 (main text).}     
	\label{Fig:Sensitivity-DL}
\end{figure}
\begin{figure}[ht]
	\captionsetup{width=1.2\textwidth}
	{\includegraphics[width=1.15\textwidth]{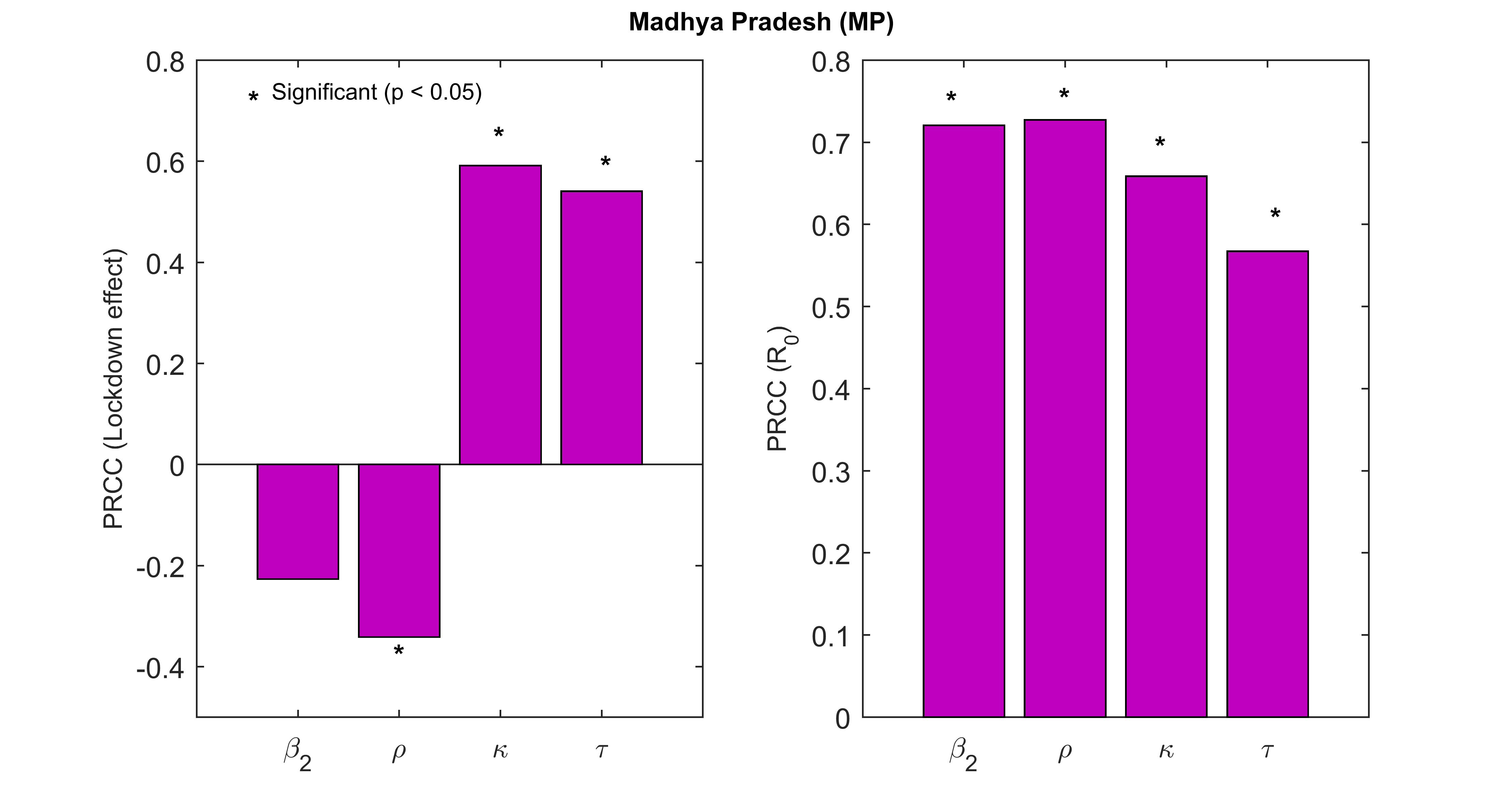}}
	\caption{Effect of uncertainty of four epidemiologically measurable and controllable parameters of the mechanistic ODE model~(see Table~1 and Fig~1 in the main text) on the effect of lockdown and the basic reproduction number ($R_{0}$). Lockdown effect is measured in terms of the differences in total number of COVID-19 cases occurred during May 3, 2020 till May 20, 2020 in \textbf{Madhya Pradesh} under the lockdown scenarios (\textbf{A1}) and (\textbf{A5}), respectively~(see main text). Effect of Uncertainty of these four parameters on the two mentioned responses are measured using Partial Rank Correlation Coefficients (PRCC). $500$~samples for each parameters were drawn using Latin hypercube sampling techniques (LHS) from their respective ranges provided in Table~1 (main text).}     
	\label{Fig:Sensitivity-MP}
\end{figure}

\begin{figure}[ht]
	\captionsetup{width=1.2\textwidth}
	{\includegraphics[width=1.15\textwidth]{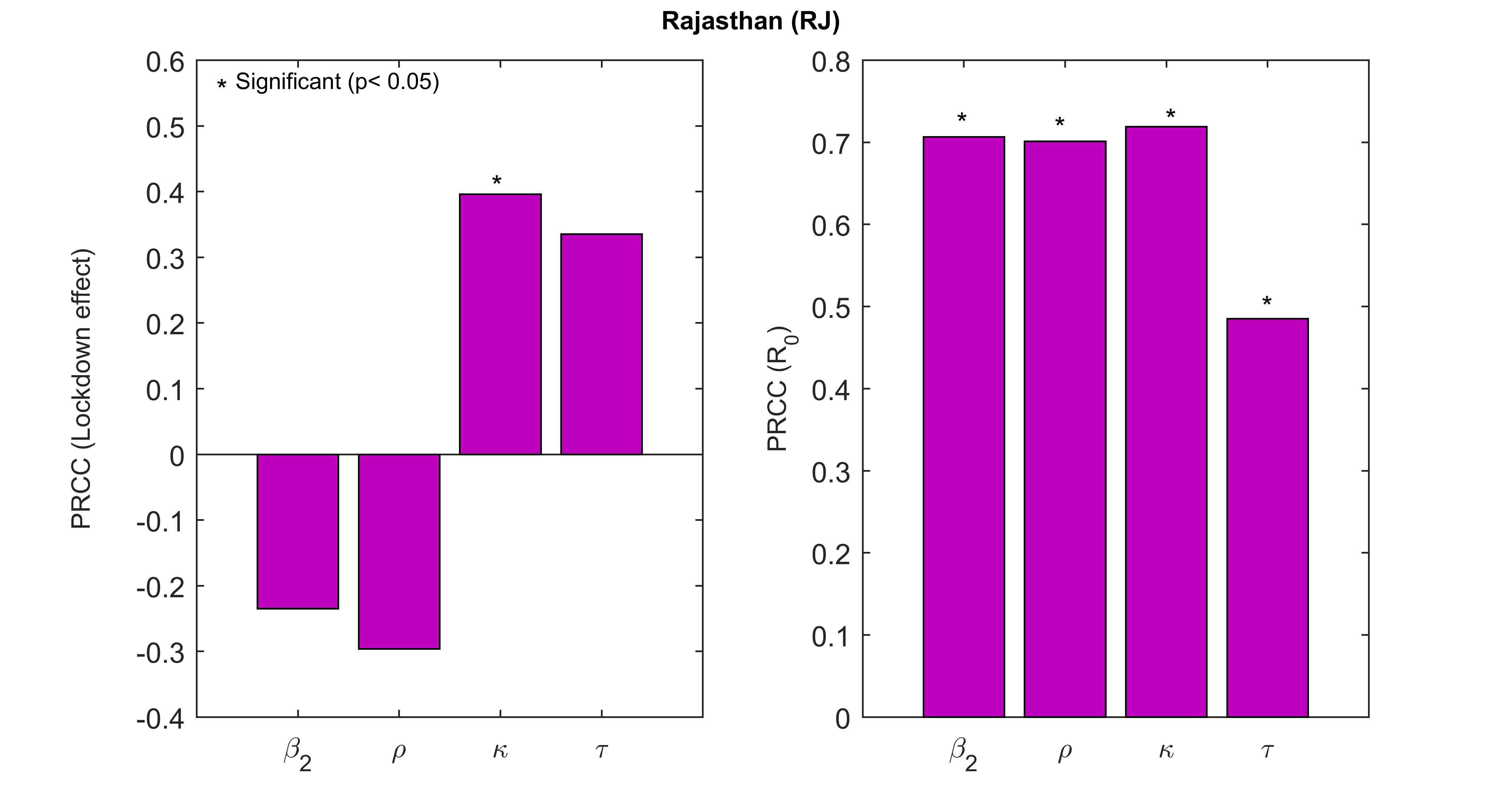}}
	\caption{Effect of uncertainty of four epidemiologically measurable and controllable parameters of the mechanistic ODE model~(see Table~1 and Fig~1 in the main text) on the effect of lockdown and the basic reproduction number ($R_{0}$). Lockdown effect is measured in terms of the differences in total number of COVID-19 cases occurred during May 3, 2020 till May 20, 2020 in \textbf{Rajasthan} under the lockdown scenarios (\textbf{A1}) and (\textbf{A5}), respectively~(see main text). Effect of Uncertainty of these four parameters on the two mentioned responses are measured using Partial Rank Correlation Coefficients (PRCC). $500$~samples for each parameters were drawn using Latin hypercube sampling techniques (LHS) from their respective ranges provided in Table~1 (main text).}     
	\label{Fig:Sensitivity-RJ}
\end{figure}

\begin{figure}[ht]
	\captionsetup{width=1.2\textwidth}
	{\includegraphics[width=1.15\textwidth]{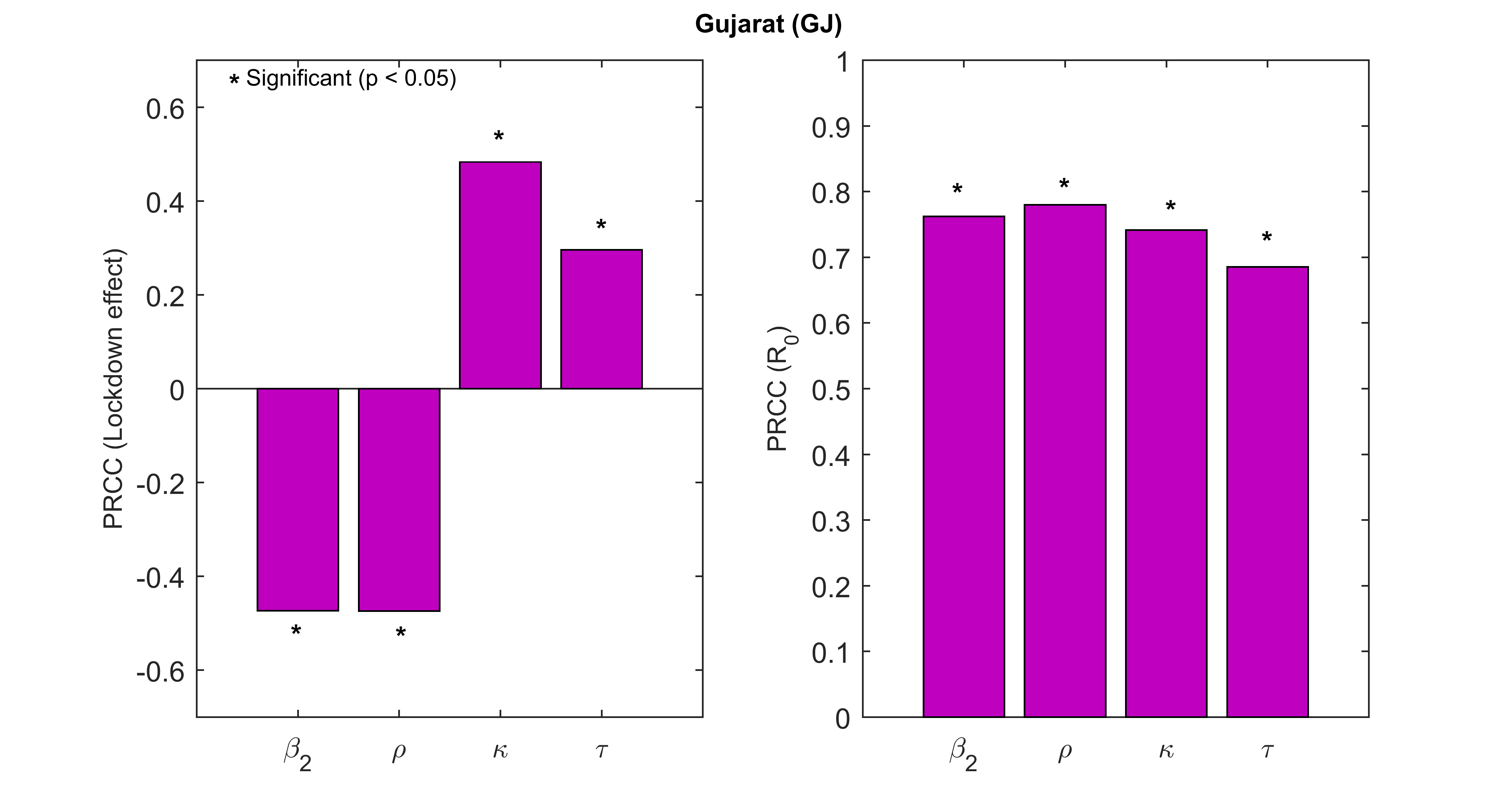}}
	\caption{Effect of uncertainty of four epidemiologically measurable and controllable parameters of the mechanistic ODE model~(see Table~1 and Fig~1 in the main text) on the effect of lockdown and the basic reproduction number ($R_{0}$). Lockdown effect is measured in terms of the differences in total number of COVID-19 cases occurred during May 3, 2020 till May 20, 2020 in \textbf{Gujarat} under the lockdown scenarios (\textbf{A1}) and (\textbf{A5}), respectively~(see main text). Effect of Uncertainty of these four parameters on the two mentioned responses are measured using Partial Rank Correlation Coefficients (PRCC). $500$~samples for each parameters were drawn using Latin hypercube sampling techniques (LHS) from their respective ranges provided in Table~1 (main text).}     
	\label{Fig:Sensitivity-GJ}
\end{figure}

\begin{figure}[ht]
	\captionsetup{width=1.2\textwidth}
	{\includegraphics[width=1.15\textwidth]{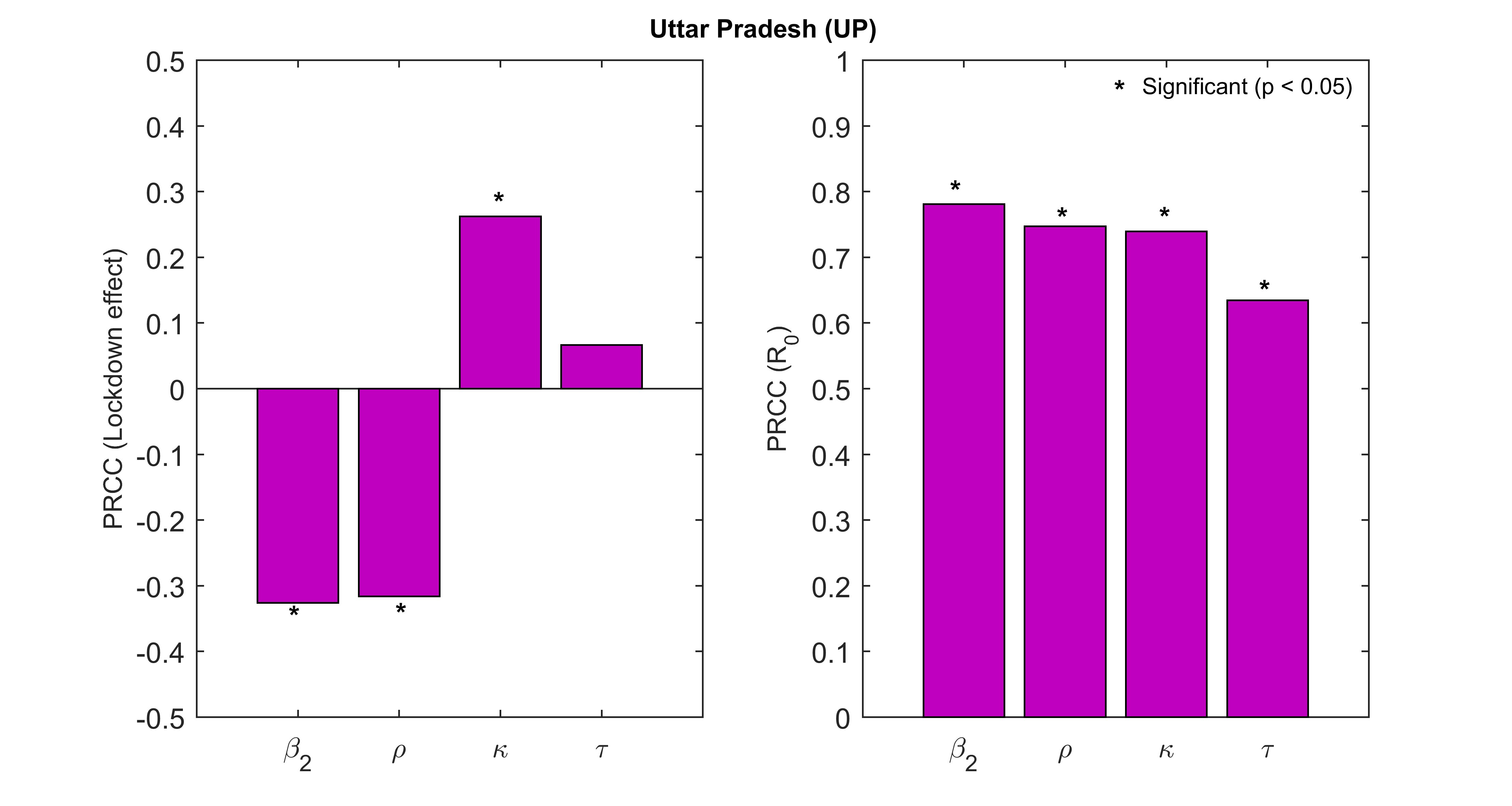}}
	\caption{Effect of uncertainty of four epidemiologically measurable and controllable parameters of the mechanistic ODE model~(see Table~1 and Fig~1 in the main text) on the effect of lockdown and the basic reproduction number ($R_{0}$). Lockdown effect is measured in terms of the differences in total number of COVID-19 cases occurred during May 3, 2020 till May 20, 2020 in \textbf{Uttar Pradesh} under the lockdown scenarios (\textbf{A1}) and (\textbf{A5}), respectively~(see main text). Effect of Uncertainty of these four parameters on the two mentioned responses are measured using Partial Rank Correlation Coefficients (PRCC). $500$~samples for each parameters were drawn using Latin hypercube sampling techniques (LHS) from their respective ranges provided in Table~1 (main text).}     
	\label{Fig:Sensitivity-UP}
\end{figure}

\begin{figure}[ht]
	\begin{center}
		\includegraphics[width=0.8\textwidth]{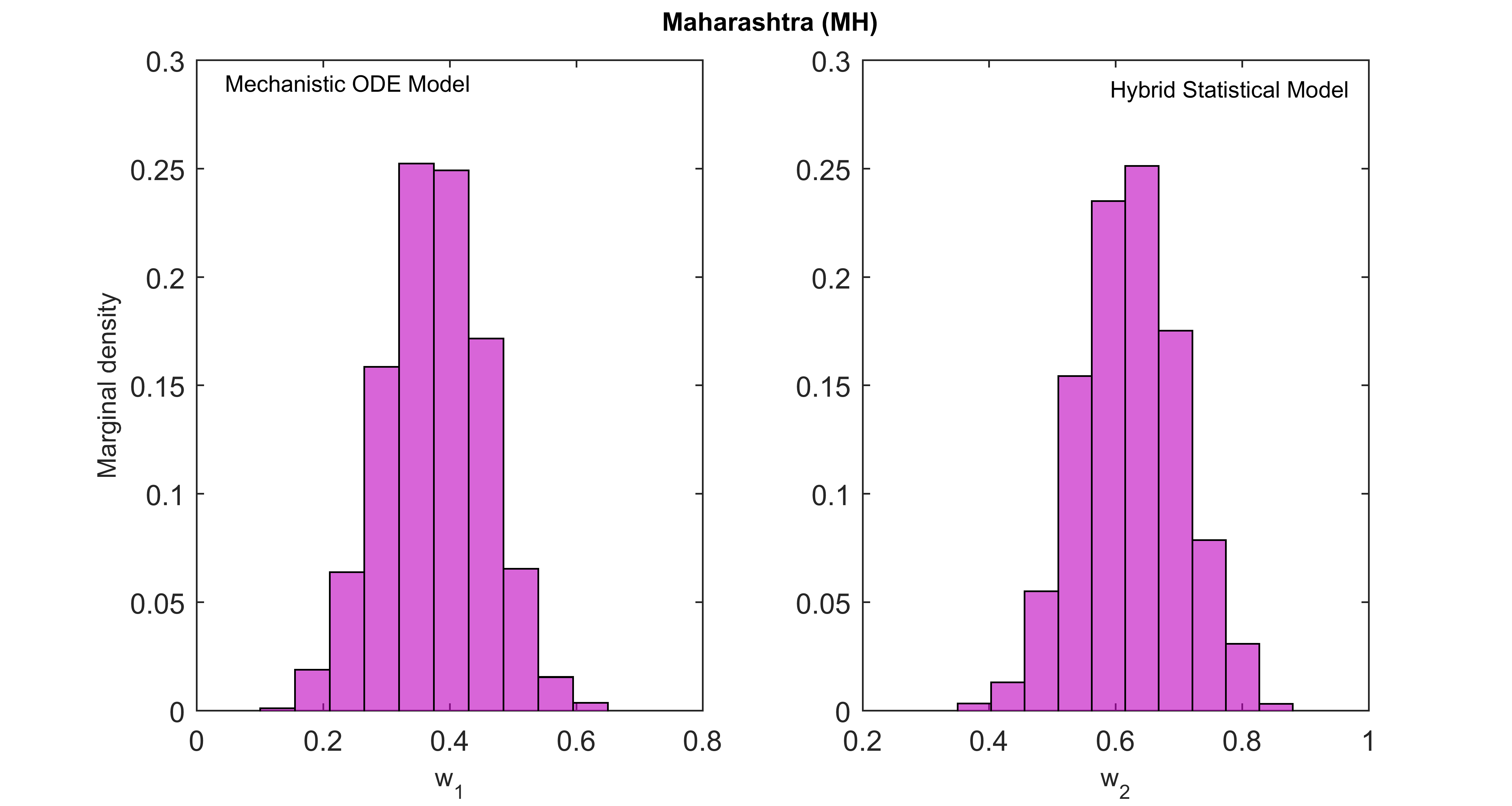}
	\end{center}
	\caption{Posterior distribution of the weights for the mechanistic ODE model combinations~(\ref{EQ:eqn 2.1}) \&~(\ref{EQ:eqn 2.2}) and the Hybrid statistical model (see main text), respectively for Maharashtra.}
	\label{Fig:posterior-weight-MH}
\end{figure}
\begin{figure}
	\begin{center}
		\includegraphics[width=0.8\textwidth]{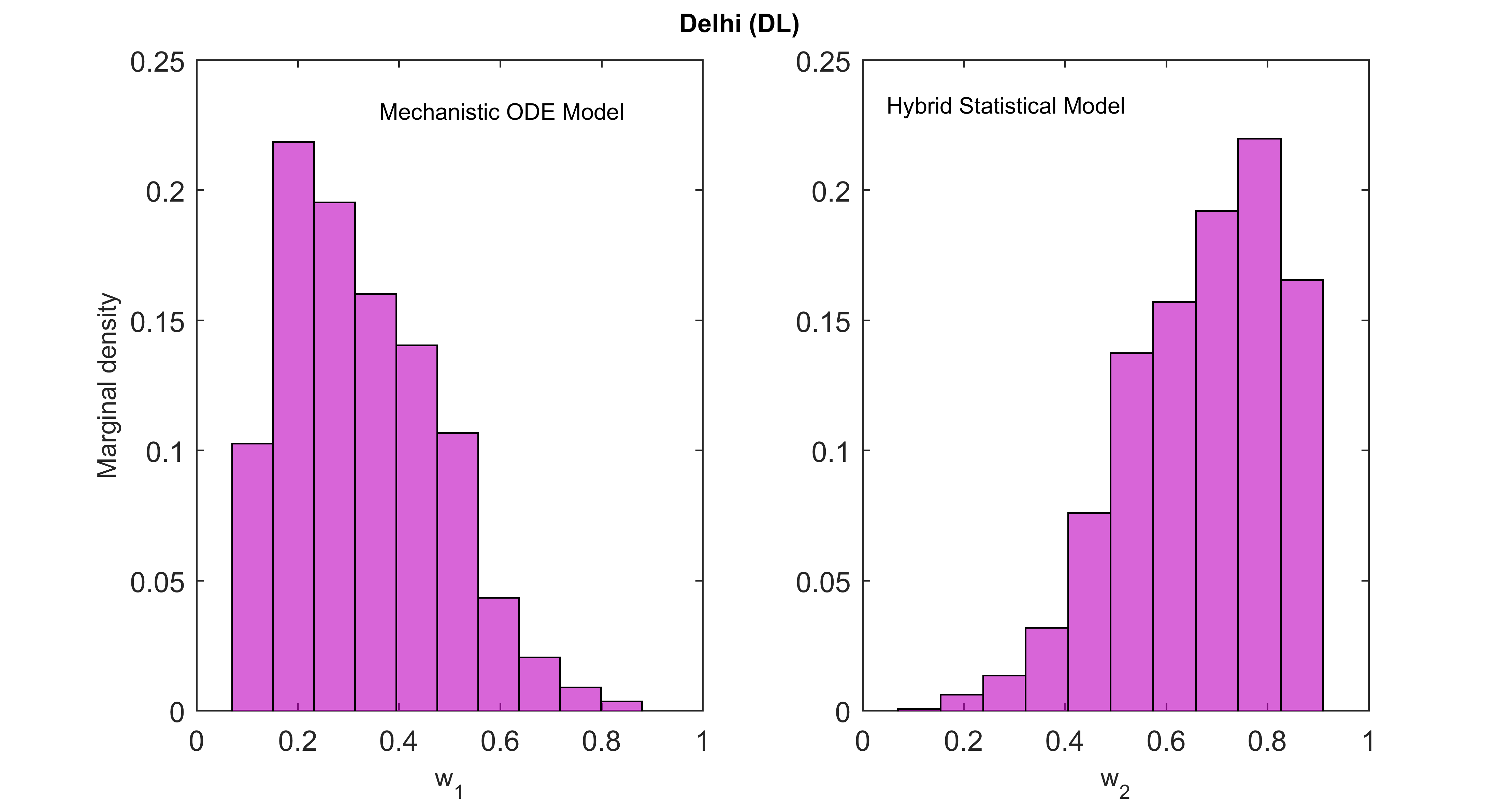}
	\end{center}
	\caption{Posterior distribution of the weights for the mechanistic ODE model combinations~(\ref{EQ:eqn 2.1}) \&~(\ref{EQ:eqn 2.2}) and the Hybrid statistical model (see main text), respectively for Delhi.}
	\label{Fig:posterior-weight-DL}
\end{figure}

\begin{figure}
	\begin{center}
		\includegraphics[width=0.8\textwidth]{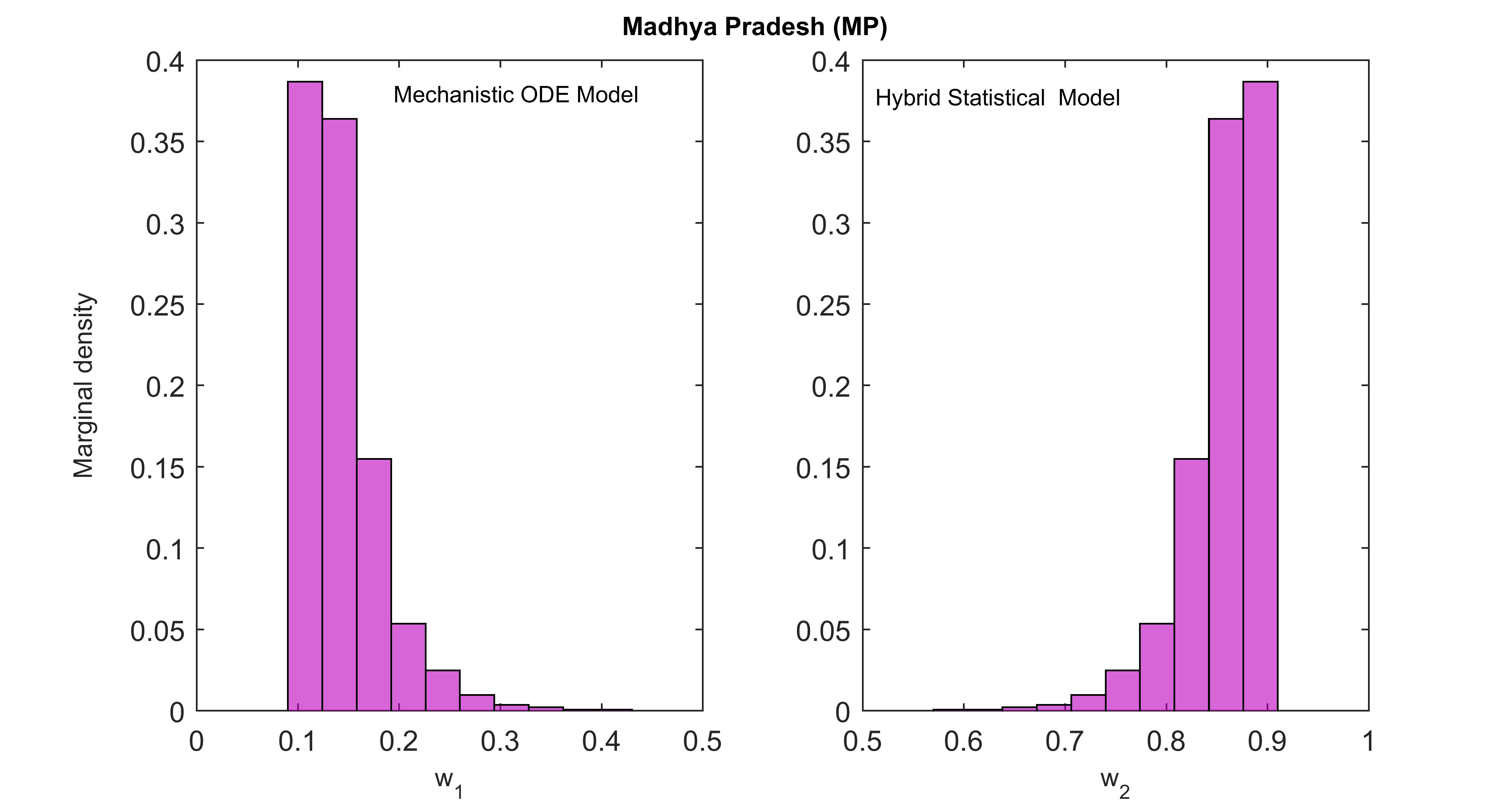}
	\end{center}
	\caption{Posterior distribution of the weights for the mechanistic ODE model combinations~(\ref{EQ:eqn 2.1}) \&~(\ref{EQ:eqn 2.2}) and the Hybrid statistical model (see main text), respectively for Madhya Pradesh.}
	\label{Fig:posterior-weight-MP}
\end{figure}
\begin{figure}
	\begin{center}
		\includegraphics[width=0.8\textwidth]{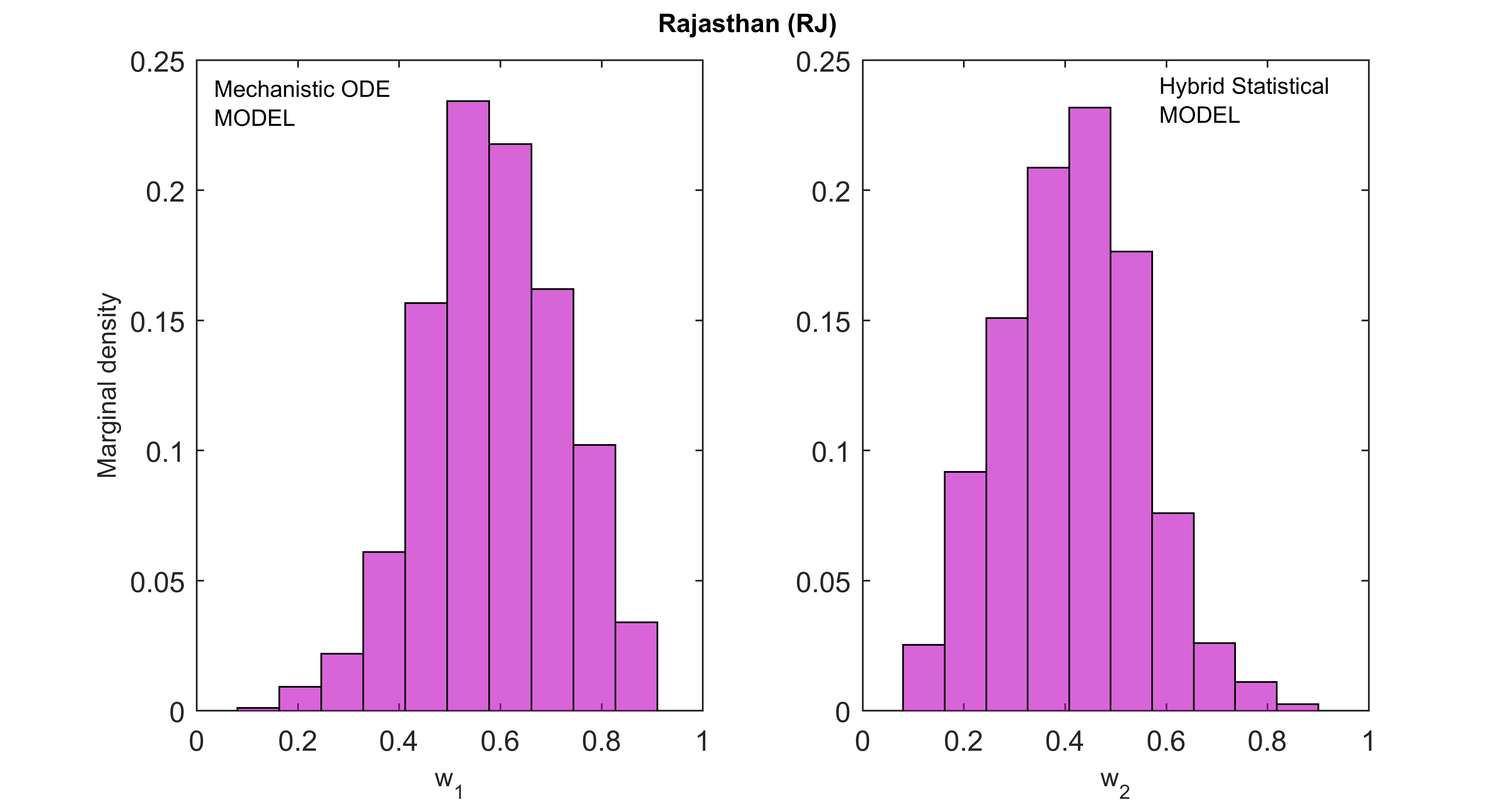}
	\end{center}
	\caption{Posterior distribution of the weights for the mechanistic ODE model combinations~(\ref{EQ:eqn 2.1}) \&~(\ref{EQ:eqn 2.2}) and the Hybrid statistical model (see main text), respectively for Rajasthan.}
	\label{Fig:posterior-weight-RJ}
\end{figure}

\begin{figure}
	\begin{center}
		\includegraphics[width=0.8\textwidth]{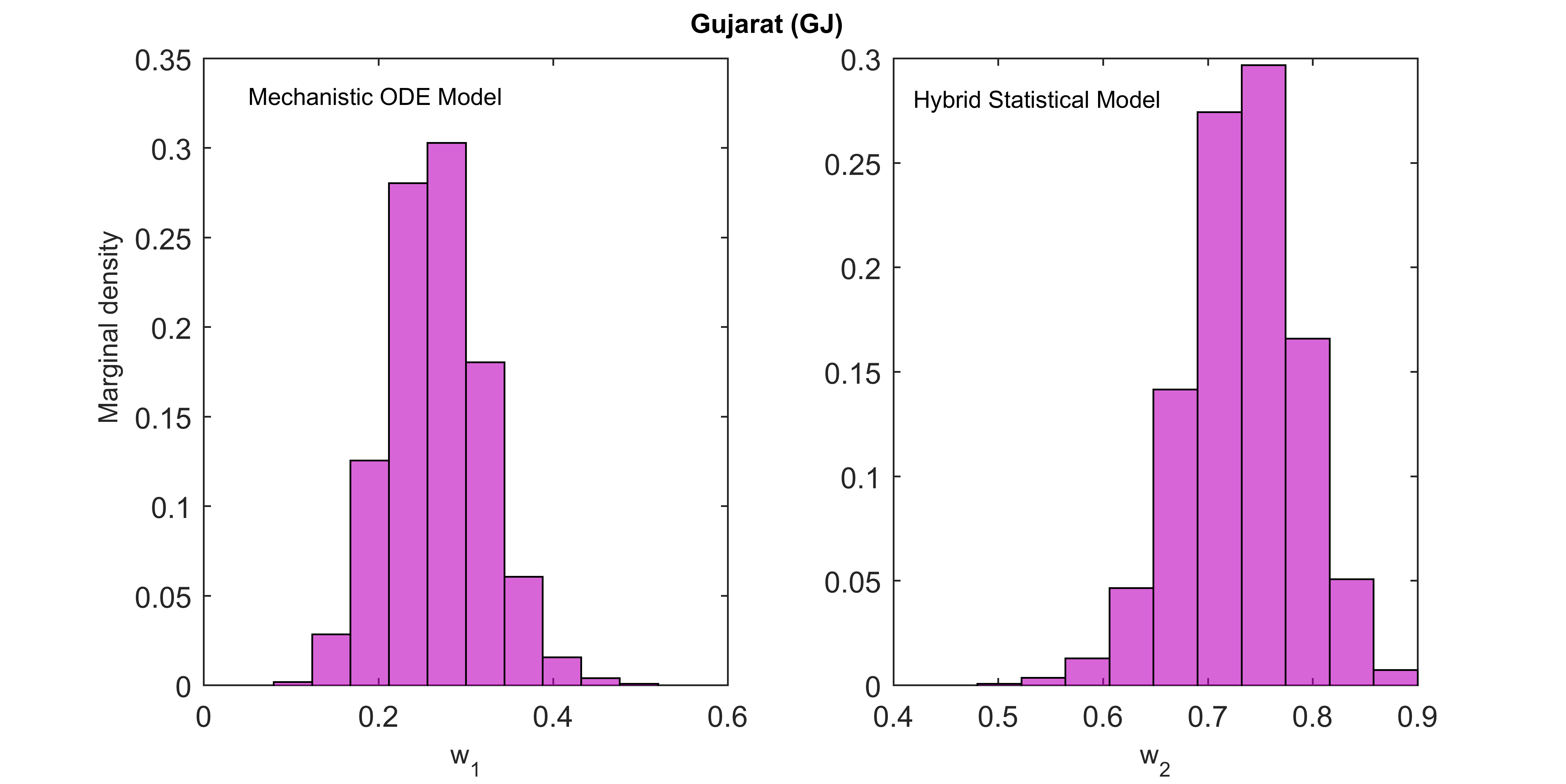}
	\end{center}
	\caption{Posterior distribution of the weights for the mechanistic ODE model combinations~(\ref{EQ:eqn 2.1}) \&~(\ref{EQ:eqn 2.2}) and the Hybrid statistical model (see main text), respectively for Gujarat.}
	\label{Fig:posterior-weight-GJ}
\end{figure}

\begin{figure}
	\begin{center}
		\includegraphics[width=0.8\textwidth]{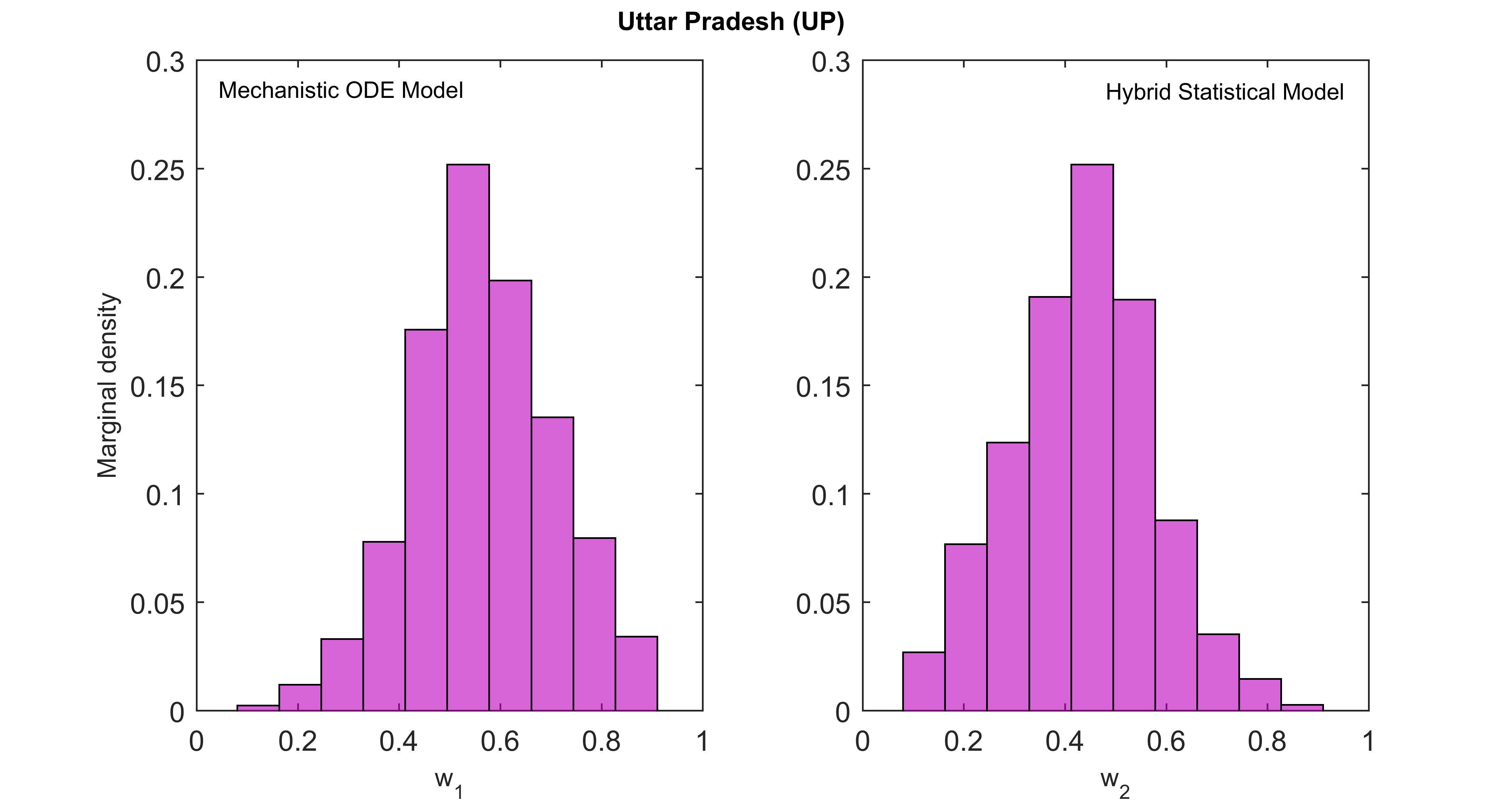}
	\end{center}
	\caption{Posterior distribution of the weights for the mechanistic ODE model combinations~(\ref{EQ:eqn 2.1}) \&~(\ref{EQ:eqn 2.2}) and the Hybrid statistical model (see main text), respectively for Uttar Pradesh.}
	\label{Fig:posterior-weight-UP}
\end{figure}

\begin{figure}
	\begin{center}
		\includegraphics[width=0.8\textwidth]{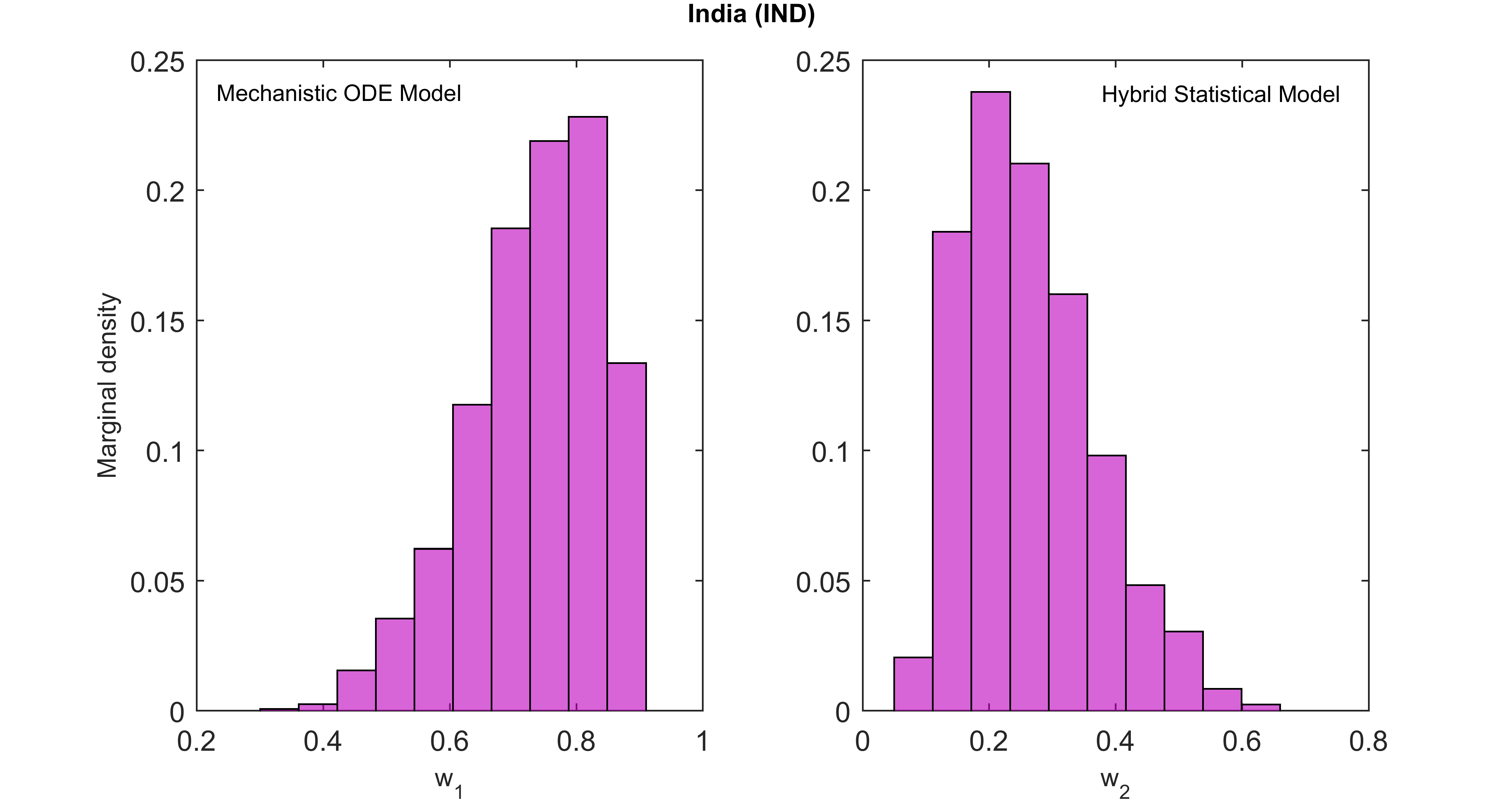}
	\end{center}
	\caption{Posterior distribution of the weights for the mechanistic ODE model combinations~(\ref{EQ:eqn 2.1}) \&~(\ref{EQ:eqn 2.2}) and the Hybrid statistical model (see main text), respectively for India.}
	\label{Fig:posterior-weight-IND}
\end{figure}
\clearpage

\begin{center}
	\section*{\Large{\underline{Tables}}}
\end{center}

\begin{table}[ht]
	\captionsetup{font=normalsize}
	\captionsetup{width=1.1\textwidth}
	\tabcolsep 19pt		
	\centering
	\caption{Estimated uninformative initial conditions for the mechanistic ODE model~(\ref{EQ:eqn 2.1}). Respective subscripts are \textbf{MH}:~Maharashtra, \textbf{DL}:~Delhi, \textbf{MP}:~Madhya Pradesh, \textbf{RJ}:~Rajasthan, \textbf{GJ}:~Gujarat, \textbf{UP}:~Uttar Pradesh, and \textbf{IND}:~India. All data are provided in the format~\textbf{Estimate (95\% CI)}.}\vspace{0.2cm}
	\begin{tabular}{|p{2.4cm}|ccccccccc} \hline\\
		\textbf{Region} &  $\boldsymbol{S(0)}$ & $\boldsymbol{E(0)}$  & $\boldsymbol{I(0)}$  & $\boldsymbol{A(0)}$ \\ \hline\\
		\textbf{MH}& $\substack{120428468  \\  \\ (101502357 - 124478508)}$ & $\substack{0.001  \\  \\ (0 - 0.003)}$ & $\substack{332.43  \\  \\ (6.14 - 633.93)}$ & $\substack{2316.93  \\  \\ (230.8 - 7103.05)}$\\\\
		\hline\\
		\textbf{DL}& $\substack{12173507  \\  \\ (10289760 - 19705021)}$ & $\substack{3743  \\  \\ (179- 9704)}$ & $\substack{0.14  \\  \\ (0 - 0.29)}$ & $\substack{245  \\  \\ (144.5 - 9409)}$\\\\
		\hline\\
		\textbf{MP}& $\substack{84632517  \\  \\ (70665560 - 89661041)}$ & $\substack{1541 \\  \\ (65.96 - 4956)}$ & $\substack{0.2914 \\  \\ (0.01 - 1.10)}$ & $\substack{4095  \\  \\ (93.46 - 7852)}$\\\\
		\hline\\
		\textbf{RJ}& $\substack{74502407  \\  \\ (70099616 - 79723261)}$ & $\substack{195.9 \\  \\ (43.74 - 9770)}$ & $\substack{129  \\  \\ (0.7408 - 132)}$ & $\substack{6209  \\  \\ (40.86 - 9404)}$\\\\
		\hline\\
		\textbf{GJ}& $\substack{62744520  \\  \\ (60628093 - 69846136)}$ & $\substack{1280 \\  \\ (58.78 - 2238)}$ & $\substack{79.29  \\  \\ (3.70 - 237)}$ & $\substack{1.17  \\  \\ (0.39 - 2.60)}$\\\\
		\hline\\
		\textbf{UP}& $\substack{224225631  \\  \\ (200822802 - 229090013)}$ & $\substack{0.49 \\  \\ (0.01 - 1.34)}$ & $\substack{3.07  \\  \\ (0.17 - 7.35)}$ & $\substack{4764  \\  \\ (379 - 8031)}$\\\\
		\hline\\
		\textbf{IND}& $\substack{1219512468  \\  \\ (1219512465 - 1219512469)}$ & $\substack{49998 \\  \\ (49995 - 49999)}$ & $\substack{9997  \\  \\ (9994 - 9999)}$ & $\substack{9995  \\  \\ (9991 - 9999)}$\\\\
		\hline
	\end{tabular}
	\label{Tab:estimated-initial-Table}
\end{table}

\begin{table}[ht]
	\captionsetup{font=normalsize}
	\captionsetup{width=1.1\textwidth}
	\tabcolsep 4.5pt		
	\centering
	\caption{Weight estimates for the mechanistic ODE model combinations~(\ref{EQ:eqn 2.1}) \&~(\ref{EQ:eqn 2.2}) and the Hybrid statistical model (see main text), respectively. Respective subscripts \textbf{MH}, \textbf{DL}, \textbf{MP}, \textbf{RJ}, \textbf{GJ}, \textbf{UP}, and \textbf{IND} are same as Table~\ref{Tab:estimated-initial-Table}. All data are provided in the format~\textbf{Estimate (95\% CI)}.}  
	\begin{tabular}{|p{1.8cm}|ccccccccc} \hline\\
		\textbf{Weights} &  \textbf{MH} & \textbf{DL}  & \textbf{MP}  & \textbf{RJ} &  \textbf{GJ}  & \textbf{UP} & \textbf{IND} \\ \hline\\
		$\boldsymbol{w_{1}}$& \footnotesize{$\substack{0.3019  \\  \\ (0.2160 - 0.5284)}$} & \footnotesize{$\substack{0.3226  \\  \\ (0.1149 - 0.6594)}$} & \footnotesize{$\substack{0.1599  \\  \\ (0.0966 - 0.2384)}$} & \footnotesize{$\substack{0.5490  \\  \\ (0.3091 - 0.8384)}$} & \footnotesize{$\substack{0.1672  \\  \\ (0.1665 - 0.38)}$} & \footnotesize{$\substack{0.6612  \\  \\ (0.2807 - 0.8333)}$} & \footnotesize{$\substack{0.8176  \\  \\ (0.5017 - 0.8874)}$}\\\\
		\hline\\
		$\boldsymbol{w_{2}}$ & \footnotesize{$\substack{0.6981  \\  \\ (0.4716-0.7840)}$} & \footnotesize{$\substack{0.6774  \\  \\ (0.3405 - 0.8851)}$} & \footnotesize{$\substack{0.8401 \\  \\ (0.7616-0.9034)}$} & \footnotesize{$\substack{0.4510  \\  \\ (0.1616 - 0.6909)}$} & \footnotesize{$\substack{0.8328  \\  \\ (0.62 - 0.8335)}$} & \footnotesize{$\substack{0.3388  \\  \\ (0.1667 - 0.7193)}$} & \footnotesize{$\substack{0.1824  \\  \\ (0.1126 - 0.4983)}$}\\
		\hline		
	\end{tabular}
	\label{Tab:estimated-weights}
\end{table}
\begin{table}[ht]
	\captionsetup{font=normalsize}
	\captionsetup{width=1.1\textwidth}
	\tabcolsep 13pt		
	\centering
	\caption{Goodness of fit (RMSE and MAE) of the Hybrid statistical model (see main text) for the test data from \textbf{MH}, \textbf{DL}, \textbf{MP}, \textbf{RJ}, \textbf{GJ}, \textbf{UP}, and \textbf{IND}, respectively. Respective subscripts \textbf{MH}, \textbf{DL}, \textbf{MP}, \textbf{RJ}, \textbf{GJ}, \textbf{UP}, and \textbf{IND} are same as Table~\ref{Tab:estimated-initial-Table}.}  
	\begin{tabular}{|p{3.8cm}|ccccccc} \hline\\
		\textbf{Goodness of fit} &  \textbf{MH} & \textbf{DL}  & \textbf{MP}  & \textbf{RJ} &  \textbf{GJ}  & \textbf{UP} & \textbf{IND} \\ \hline\\
		$\boldsymbol{RMSE}$& \footnotesize{$\substack{35.15368}$} & \footnotesize{$\substack{57.855}$} & \footnotesize{$\substack{28.59711}$} & \footnotesize{$\substack{18.4018}$} & \footnotesize{$\substack{21.77785}$} & \footnotesize{$\substack{17.60281}$} & \footnotesize{$\substack{114.3764}$}\\\\
		\hline\\
		$\boldsymbol{MAE}$ & \footnotesize{$\substack{22.2375}$} & \footnotesize{$\substack{34.4592}$} & \footnotesize{$\substack{20.01179}$} & \footnotesize{$\substack{13.22069}$} & \footnotesize{$\substack{12.91013}$} & \footnotesize{$\substack{10.16784}$} & \footnotesize{$\substack{65.61564}$}\\
		\hline		
	\end{tabular}
	\label{Tab:rmse-mae}
\end{table}

\end{document}